\documentclass[aps,pra,twocolumn,superscriptaddress,nofootinbib]{revtex4-2}%
\usepackage{CJK}
\usepackage{physics}
\usepackage{amstext}
\usepackage{amsfonts,amssymb}
\usepackage{bbm}
\usepackage{graphicx}
\usepackage{float}
\usepackage{indentfirst}
\usepackage{geometry}
\usepackage{mathrsfs}
\usepackage[normalem]{ulem}
\usepackage{color}
\usepackage{amsthm}
\usepackage{graphicx}
\usepackage{dcolumn}
\usepackage{bm}
\usepackage[utf8]{inputenc}
\usepackage[T1]{fontenc}
\usepackage{booktabs, array, mathptmx, lipsum, amsmath,multirow}
\usepackage{siunitx, xcolor}
\usepackage[version=4]{mhchem}
\usepackage{thmtools}
\usepackage{tikz}
\newtheorem{theorem}{Theorem}
\newtheorem{lemma}{\normalsize Lemma}
\theoremstyle{definition}
\newtheorem{assume}{Postulate}
\newtheorem{define}{Definition}
\usepackage[colorlinks,linkcolor=blue,anchorcolor=blue,citecolor=blue,
     urlcolor=black]{hyperref}
\usepackage{nomencl}
\def\6{{\langle}}
\def\9{{\rangle}}

\def\pad{{\partial}}

\def\vH  {{\boldsymbol{\xi}_H}}
\def\veta{{\boldsymbol{\eta}}}

\newcommand{\be}{\begin{equation}}
\newcommand{\ee}{\end{equation}}
\newcommand{\ba}{\begin{eqnarray}}
\newcommand{\ea}{\end{eqnarray}}
\def\gP{\mathfrak{P}}
\def\gD{\mathfrak{D}}

\newcommand{\defeq}{\vcentcolon=}

\usepackage{epsf,latexsym,bbm,euscript}
\usepackage{mathtools} 
\usepackage{soul,xcolor}
\usepackage{mathtools}
\usepackage{mathrsfs}

\usepackage{enumitem}

\begin{document}
\title{From reasonable postulates to generalised Hamiltonian systems}
\author{Libo Jiang }
\email{11930020@mail.sustech.edu.cn}
\affiliation{Alumnus,
Southern University of Science and Technology (SUSTech), Shenzhen 518055, China}

\pacs{123 }

\author{Daniel R. Terno}
\email{daniel.terno@mq.edu.au}
\affiliation{Department of Physics and Astronomy, Macquarie University, Sydney, New South Wales 2109, Australia}
\pacs{23 }

\author{Oscar Dahlsten}
\email{oscar.dahlsten@cityu.edu.hk}
\affiliation{Department of Physics, City University of Hong Kong, Tat Chee Avenue, Kowloon, Hong Kong SAR, China
}%
\affiliation{Shenzhen Institute for Quantum Science and Engineering and Department of Physics,
Southern University of Science and Technology, Shenzhen 518055, China}

\affiliation{Institute of Nanoscience and Applications, Southern University of Science and Technology, Shenzhen 518055, China}
\pacs{12 }
\begin{abstract}
Hamiltonian mechanics describes the evolution of a system through its Hamiltonian. The Hamiltonian typically also represents the energy observable, a Noether-conserved quantity associated with the time-invariance of the law of evolution. In both quantum and classical mechanics, Hamiltonian mechanics demands a precise relationship between time evolution and observable energy, albeit using slightly different terminology. We distil basic conditions satisfied in both quantum and classical mechanics, including canonical coordinate symmetries and inner product invariance. We express these conditions in the framework of generalised probabilistic theories, which includes generalizing the definition of energy eigenstates in terms of time-invariant properties of the Hamiltonian system. By postulating these conditions to hold, we derive a unified Hamiltonian system model. This unified framework describes quantum and classical mechanics in a consistent language, facilitating their comparison. We moreover discuss alternative theories: an equation of motion given by a mixture of commutation relations, an information-restricted version of quantum theory, and Spekken's toy theory. The findings give a deeper understanding of the Hamiltonian in quantum and classical theories and point to several potential research topics.
\end{abstract}
\maketitle
\tableofcontents

\begin{widetext}
    \begin{figure}\centering
\begin{tikzpicture}

\node[draw] (1) at(0,0) {Evolution equation};
\node[draw] (2) at(5,0) {~Hamiltonian\textcolor{white}{y}};
\node[draw] (3) at(13,0) {Pure stationary state};
\node[draw] (4) at(0,2) {Non-localized dynamics kernel $K(k)$};
\node[draw] (5) at(9,2) {Energy measurement};

\node[draw] (6) at(3,6) {Planck's constant};
\node[draw] (7) at(9,6) {State/effect volume $V$};

\draw[->] (1)--node[above]  {Dependence}(2);
\draw[->] (2)--node[above]  {Dependence}(3);
\draw[->] (1)--node[right]  {Dependence}(4);
\draw[->] (2)--node[right]  {~~One role\textcolor{white}{p}}(5);
\draw[->] (5)--node[left]  {Dependence~~}(3);
\draw[->] (4)--node[right]  {~Special case}(6);
\draw[->] (7)--node[below]  {Special case}(6);
\draw[->] (2)--node[right]  {Dependence}(7);
\draw[->] (3)--node[left]  {Property}(7);
\draw[->] (5)--node[above]  {Dependence}(7);
\end{tikzpicture}
 \caption{The relations between the key concepts in this paper.}
\end{figure}
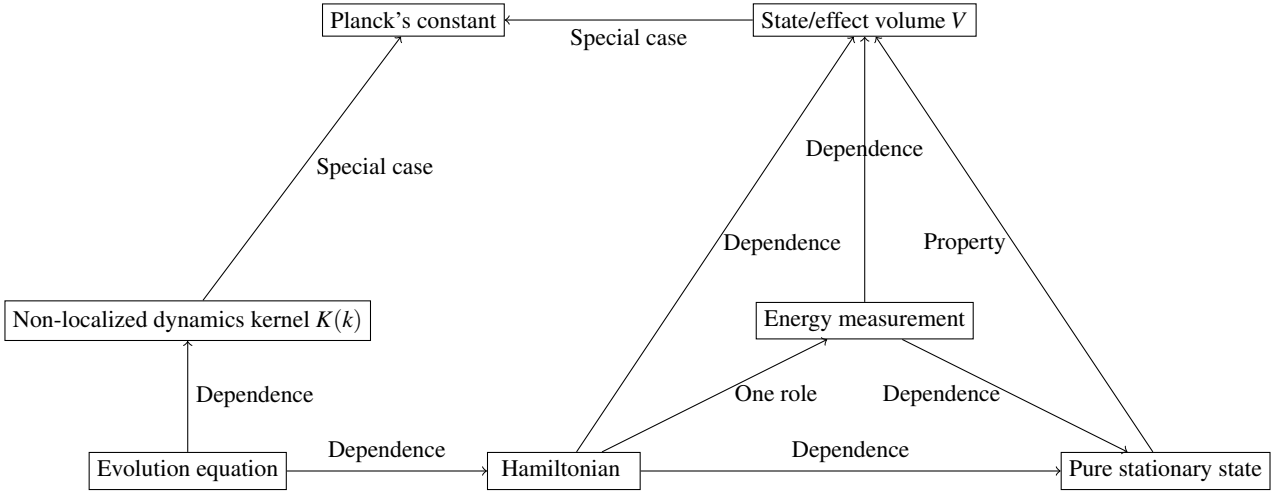
\end{widetext}

\section{Introduction}
Hamiltonian mechanics, whether in classical or quantum cases, describes the time evolution of systems through their Hamiltonian \cite{Arnold,AKN:06,L:17},
\be
\pdv{\rho(q,p)}{t}=\{H,\rho\}, \qquad  \pdv{\hat \rho}{t}=\{\hat H, \hat \rho\}_\hbar.
\ee
Here $H$ is the classical Hamiltonian and $\hat{H}$ its quantum counterpart. The classical Liouville density $\rho$ and the quantum mechanical density operator $\hat{\rho}$ evolve via their Poisson bracket with the Hamiltonian.  In the quantum case the bracket is defined as $\{\hat{A},\hat{B}\}_\hbar\defeq [\hat{A},\hat{B}]/(i\hbar)$ \cite{L:17}.  The time-independent Hamiltonian represents energy, which is a conserved dynamical quantity/quantum observable corresponding to time translation symmetry according to Noether's theorem~\cite{918Noether}. Therefore, the Hamiltonian formalism establishes a connection between energy and the time evolution.

Work in the foundations of quantum mechanics led to the development of the theoretical framework of generalised probability theories (GPTs). Systems' states in a GPT are represented as probability vectors (or other real vectors) from which probabilities can be extracted~\cite{001FivAxi, 007GPT,014GPT,012RevDua,013UncDyn}. The state depends on the system preparation and any subsequent dynamical transformations.  A key goal of GPTs is to understand the structure of quantum theory, particularly which elements necessarily follow from its probabilistic nature, and to elucidate the relations between classical and quantum mechanics~\cite{001FivAxi,007GPT,014GPT,P:23}. Classical and quantum theories, as well as classical-quantum hybrid models \cite{013HybDyn,T:23}, appear as special cases.

Using the GPT framework, the notion of Hamiltonian has been generalised in finite dimensions~\cite{barnum2014higher, 018GenHam}. There is also a long-running interest in unifying and comparing quantum and classical mechanics~\cite{T:23,017MaLeAd}, as well as efforts to explore potential new theories via the creation of toy theories that are not classical or quantum~\cite{012SpeToy, PlavalaK22, 022ToogeneralEOM}. Taken together, these results give hope that a more fully generalised Hamiltonian mechanics can be created, encompassing classical and quantum mechanics and more.

Here, we accordingly aim to create a framework for generalised Hamiltonian systems, giving full details of the results in Ref.~\cite{shortversion}. These are foundational efforts, strengthening our understanding of quantum and classical mechanics and what may lie beyond.

We generalise the phase space representation of quantum mechanics~\cite{995PhaDis,ZFC:05}, specifically the Wigner function~\cite{ZFC:05,932WigFun,984WigRev,977WigPar}. The phase space formalism models quantum and classical mechanics in a similar manner. Quantum states are described by real quasi-probability distributions of position and momentum $W(q,p)$. We make use of the fact that the Wigner function framework may be viewed as a continuous variable generalised probabilistic theory~\cite{016TunNeg}.

We establish a Hamiltonian formalism for GPTs based on postulates that are satisfied by both quantum and classical theories.   Through these postulates, we obtain a generalised measurement of energy and a generalised equation of motion. A specific theory, such as quantum or classical mechanics, is obtained by completing the system of axioms in a way that is described below.

A fundamental aspect of this framework is the observation that both quantum and classical evolutions can be generated by pure stationary states, which then serve as the generalised energy eigenstates within our framework. We define these to be the most pure stationary states, a definition that leads to simple expressions for the generalised mechanics. The {set of} stationary states is a time-independent characteristic of the system and encodes the key part of the time-independent evolution rule.

The Planck constant plays a crucial role in distinguishing between quantum and classical theories. In the generalised framework, we find that $\hbar$ has two distinct roles. One role pertains to the uncertainty of the state, which we refer to as the \textit{state/effect volume}, while the other role appears in the equation of motion via a non-localized dynamics integration kernel.  Although these quantities coincide in quantum and classical mechanics, they may have different values in generalised theories, such as in epistemically restricted classical theory and quantum mechanics with a particular information restriction.

Furthermore, an intriguing finding beyond quantum and classical mechanics is a new equation of motion. It is given by a series of commutators with the Hamiltonian for which each commutator can be different, resulting in a non-associative algebra. This new evolution rule, derived here from reasonable postulates, happens to provide a restricted version of the `generalised Moyal bracket' proposed in Ref.~\cite{022ToogeneralEOM}.

As an application of our model, we demonstrate how the concept of state volume helps to understand the possibility of chaos in the sense of strong sensitivity to initial state perturbations~\cite{014Chaos,016Chaos}. The contrasting chaotic behaviours observed in quantum and classical cases can be attributed to differences in the volumes of pure states. Besides, despite developing the model within the framework of continuous phase space, we discover that certain concepts, such as state/effect volume, can also be extended to discrete systems, such as Spekken's toy model~\cite{012SpeToy}.

\section{Preliminaries}
In this section, we summarise some key results central to this paper from generalised probabilistic theories, the phase space formalism, and the action-angle formalism of classical mechanics.

\subsection{Phase space representation}
Classical models that we consider describe non-constrained systems with a finite number of $n$  degrees of freedom. Its states and (the algebra of) observables are smooth functions on the phase
space   $\mathfrak{P}$ ~\cite{Arnold,AKN:06,L:17}, which is then a $2n$ dimensional symplectic manifold that is a cotangent bundle of the configuration space. {(Some mathematical aspects of the phase space formalism as summarised in Appendix \ref{phs})} The local coordinates on $\mathfrak{P}$ are   $z = (q, p)$, where $q$ are the generalised coordinates and the $p$ are the canonical conjugate momenta. Pure states represent the perfect knowledge of position and momentum and are thus $\delta$-distributions, $\rho_{z_0}=\delta(q-q_0)\delta(p-p_0)$.  In situations of incomplete knowledge about a system's state, like in statistical mechanics, the state is represented as a probability (Liouville) density $\rho(q,p)$\footnote{Notice that the density matrix $\hat \rho$ has the unit of probability instead of density, while we choose $\rho$ to represent probability densities.}.

The evolution of the state is generated by the Hamiltonian $H(q,p)$ according to the Hamilton equations $\dot{q}=\pdv{H}{p}$ and $\dot{p}=-\pdv{H}{q}$. A {probability density} $\rho(q,p)$ then evolves according to the Liouville equation
\footnote{Consider time $t\rightarrow t+dt$, where $dt$ is sufficiently small. Then, under the tangential approximation, $\rho(q,p,t)\rightarrow \rho(q+\pdv{q}{t}dt,p+\pdv{p}{t}dt,t+dt)$,  $d\rho=\pdv{\rho}{q}\pdv{q}{t}dt+\pdv{\rho}{p}\pdv{p}{t}dt+\pdv{\rho}{t}dt$. By Liouville's theorem,  $d\rho=0$. Combining this with Hamilton's equations gives the Liouville equation.},
\begin{equation}
\pdv{\rho(q,p)}{t}=\pdv{\rho}{p}\pdv{H}{q}-\pdv{\rho}{q}\pdv{H}{p}=\{H,\rho\}=-H\Lambda \rho,\label{Leq}
\end{equation}
where $H$ is the Hamiltonian, and $\{~,~\}$ is  the Poisson bracket. The operator (symplectic matrix) $
\Lambda:= \overleftarrow{\frac{\partial }{\partial p}}\overrightarrow{\frac{\partial }{\partial q}}-\overleftarrow{\frac{\partial }{\partial q}}\overrightarrow{\frac{\partial }{\partial p}}
$ {provides an alternative form} of the Poisson bracket. The right/left arrows upon the operators mean that the derivative will act on the right/left side's function:  $f\overleftarrow{\frac{\partial }{\partial p}}\overrightarrow{\frac{\partial }{\partial q}}g=\pdv{f}{p}\pdv{g}{q}$.

The phase space formalism can be generalised to quantum states.
There are different approaches to quantum phase spaces~\cite{995PhaDis}. One of the most {common} versions is the Wigner function $W(q,p)$~\cite{932WigFun,984WigRev,977WigPar,ZFC:05,Peres}. Wigner functions are real functions of canonical coordinates ($q$,$p$). They are obtained  via the  Wigner transform of density matrices
\begin{equation}
\label{eq:WignerDef}
W(q,p)=\frac{1}{\pi\hbar}\int dy\expval{q-y|\hat\rho|q+y}e^{2ipy/\hbar},
\end{equation}
{with an obvious extension to $n$ degrees of freedom}.
Consider the eigenstates of a simple harmonic oscillator with $H=q^2+p^2$ as an example. Their Wigner functions are given by
\begin{equation}
W_{E_n}(q,p)=\frac{(-1)^n}{\pi\hbar}L_n\left[\frac{2(q^2+p^2)}{\hbar}\right]e^{-(q^2+p^2)/\hbar},\label{SHOES}
\end{equation}
where $L_n$ are the Laguerre polynomials. {Wigner functions are normalised,
\be
\int\!dqdpW(q,p)=1,
\ee
but not necessarily positive,}
Some example distributions are depicted in Fig. \ref{qee}.

The Born rule is reproduced by the following inner product,
\begin{equation}\mathrm{Tr}(\hat \rho_1 \hat \rho_2)=h\int W_{1}W_{2}dqdp,\label{WignerBorn}\end{equation}
where $W_1,W_{2}$ are Wigner functions corresponding to $\hat \rho_1,\hat \rho_2$ and $h$ is Planck's constant.
\begin{figure}[h]
\centering
  \includegraphics[width=8cm, trim=20 250 40 250,clip]{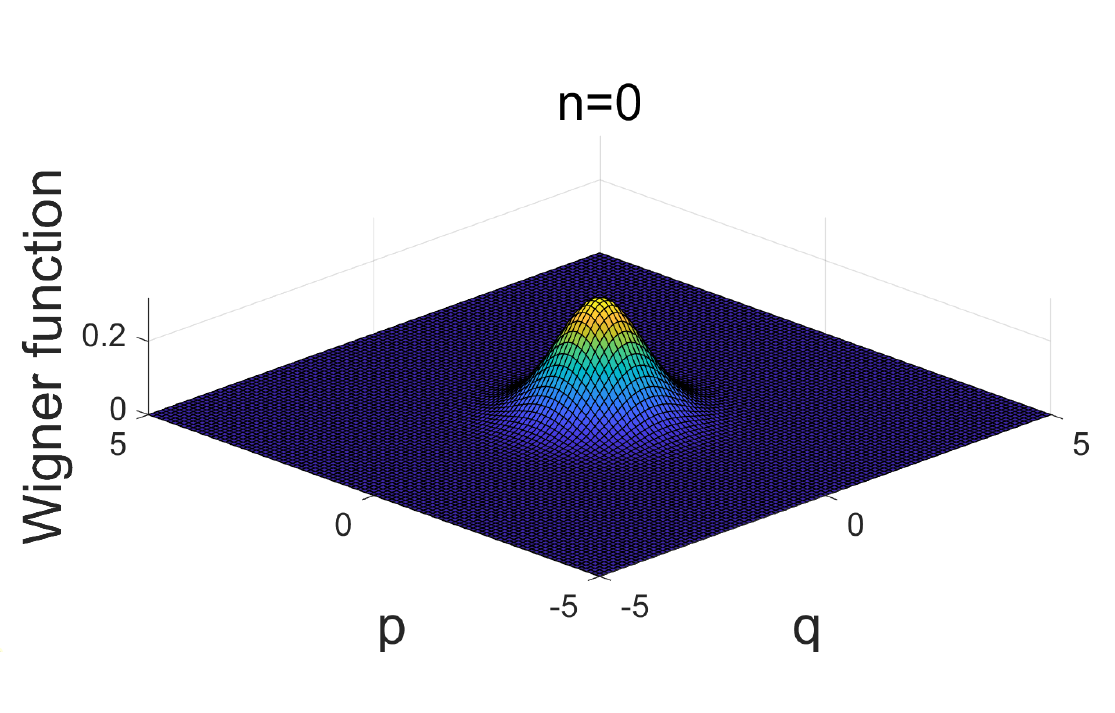}
\includegraphics[width=8cm, trim=20 250 40 250,clip]{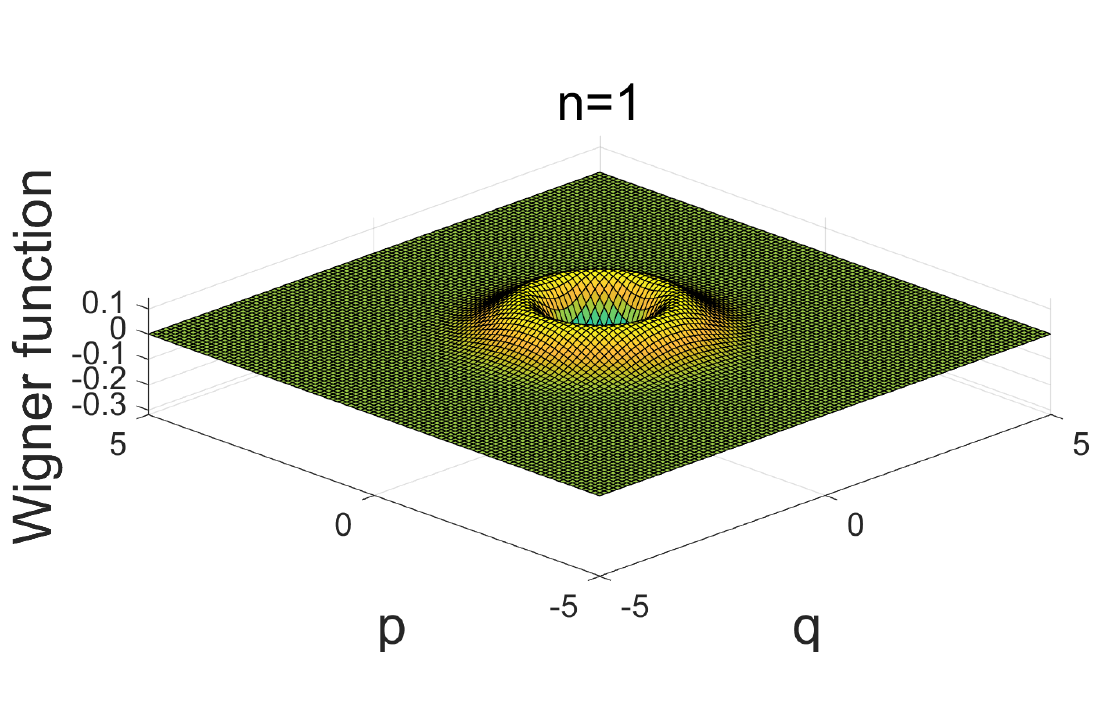} \includegraphics[width=8cm, trim=20 250 40 250,clip]{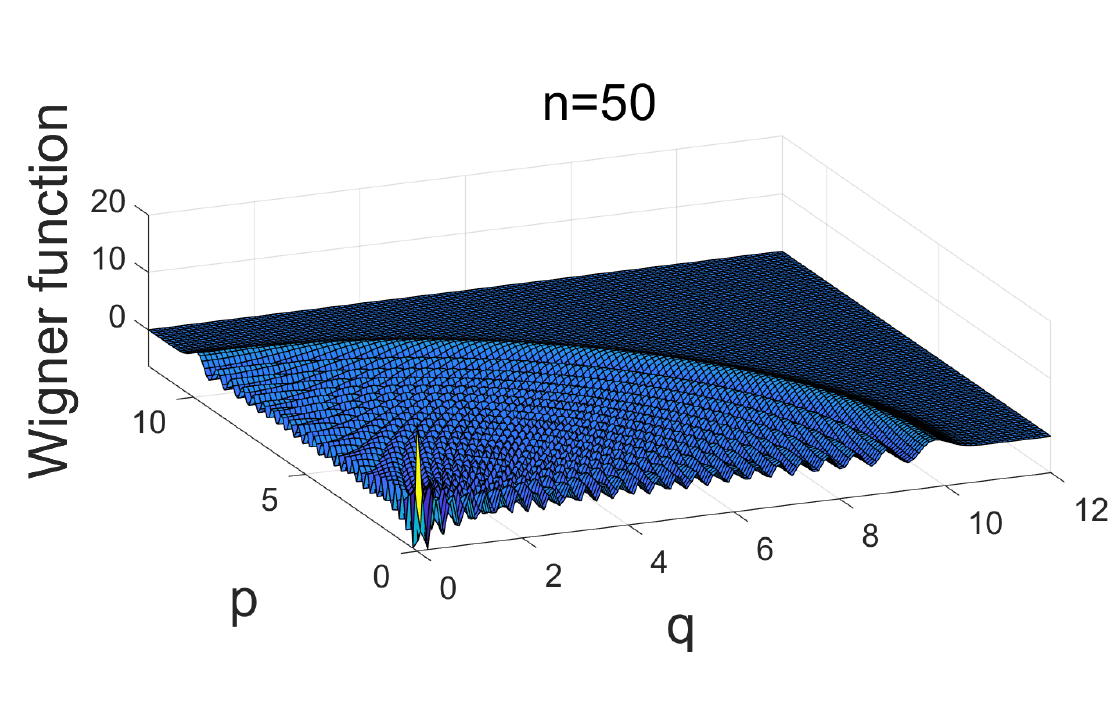}
  \caption{Wigner functions for the $n=0,1,50$ eigenstates of a simple harmonic oscillator. {$H=k(q^2+p^2)$, where $k$ is an arbitrary positive constant with the dimension of $[\frac{1}{t}]$. Both $q$ and $p$ are plotted in units of $\hbar^{\frac{1}{2}}$, and the Wigner function is plotted in units of $\frac{1}{\hbar}$.} } \label{qee}\end{figure}

 Unlike classical probability densities, many Wigner functions have some small areas with negative values and are thus called {\em quasi}-probability distribution. The probabilities of any allowed measurement outcome are nevertheless positive, which can be understood as the uncertainty principle rescuing positivity by banning measurements that would single out small phase space regions.

The inverse map that takes a phase space function to an operator is called the Weyl transform~\cite{950WeyTra}. The Weyl-Wigner transforms provide a mathematical method that connects phase space functions and non-commutative operators:
\begin{equation}\text{Wigner}\{\hat A\}(q,p)=2\int dz e^{i\frac{2pz}{\hbar}}\expval{q-z|\hat A|q+z}.\label{Wignertrans}\end{equation}
 In the other direction,
\begin{equation}\hat A= \frac{1}{4\pi^2\hbar^2}\int \text{Wigner}\{\hat A\}(q,p)e^{i\frac{a(q-\hat q)+b(p-\hat p)}{\hbar}}dqdpdadb,\label{Weyltrans}\end{equation}
where Wigner$\{~\}$ labels the Wigner transform\footnote{The Weyl-Wigner transform implies a one-to-one correspondence between real power functions  $q^np^m$ with the same power function of operators in a certain order. This order is called the Weyl ordering~\cite{984WigRev}: $q^np^m\leftrightarrow\frac{1}{2^n}\sum_{i=0}^n C_n^i \hat q^{n-i} \hat p^m \hat q^i=\frac{1}{2^m}\sum_{j=0}^m C_m^j \hat p^{m-j} \hat q^n \hat p^j$.}.

One sees the Wigner function of Eq.
\ref{Wignertrans} is the Wigner transformation of density matrices of Eq.~\eqref{eq:WignerDef} with an extra factor: $W(q,p)=\frac{1}{h}\text{Wigner}\{\hat \rho\}$. This is consistent with the Wigner function having the unit of probability density (probability per phase space area), whereas the Wigner transformation does not change the dimensionless unit of the density matrix.

The (non-commutative) product of operators appears in the Wigner function representation as
\begin{equation}
\begin{aligned}
\text{Wigner}\{\hat A \hat B\}&=\text{Wigner}\{\hat A\} \star \text{Wigner}\{ \hat B\},\\
\text{where}&\\
\star&:=\exp(-\frac{i\hbar}{2}\Lambda),\\
\end{aligned}
\end{equation}
$\star$ is the Moyal (star) product.

One can deduce the evolution of the Wigner function~\cite{984WigRev} by applying the Wigner transform and Moyal product to $i\hbar \pdv{\hat\rho}{t}=[H,\hat \rho]$, which gives:
\begin{equation}
\frac{\partial W}{\partial t}= \frac{2}{\hbar} W(q,p) \sin(\frac{\hbar}{2}\Lambda) H(q,p), \label{WFEOM}\end{equation}
where $H(q,p)$ is the Hamiltonian in phase space obtained by the Wigner transform.   When $\hbar\rightarrow 0$, Eq.\ref{WFEOM} transforms to the classical Poisson bracket of Eq.~\eqref{Leq}: $\pdv{\rho}{t}=\rho \Lambda H=-\{\rho,H\}_{P.B.}$.

The time evolution can be written in another way (the Wigner transport equation) for the case of $H=\frac{P^2}{2m}+V(q)$~\cite{984WigRev}:
\begin{equation}
\begin{array}{rl}
&\pdv{W}{t}=-\frac{p}{m}\pdv{W}{q}+\int dj W(q,p+j)J(q,j),\\
\text{where}&\\
&J(q,j)=\frac{i}{\pi\hbar^2}\int dy[V(q+y)-V(q-y)]e^{-2ijy/\hbar}.
\end{array}
\label{wignereq}
\end{equation}
The $-\frac{p}{m}\pdv{W}{q}$ term is contributed by the kinetic energy term of the Hamiltonian, while $\int dj W(q,p+j)J(q,j)$ is contributed by the potential energy term. The kinetic term is the same as in classical mechanics, but the potential term contains an integral over all momenta, implying that the distribution `jumps' in the momentum direction. The jumping in Eq.~\eqref{wignereq} is associated with the infinite orders of derivatives in Eq.~\eqref{WFEOM}, since an infinite-order Taylor expansion enables the expansion to an arbitrary distance with arbitrary precision. We will return to this point in Sec.~\ref{localnonlocal} and Appendix \ref{deltatrick}.

\subsection{Action-angle variables}\label{sec:actionangleintro}

Action-angle variables are useful in the analysis of classical systems~\cite{Arnold,AKN:06}.
The \textit{action}, also called  abbreviated action  $I$ is a number associated with an orbit defined as~\cite{Arnold}
\begin{equation}
I=\frac{1}{2\pi} \oint pdq.
\end{equation}
This quantity gives the phase space volume (enclosed by the orbit up to $\frac{1}{2\pi}$\footnote{The historically significant `Sommerfeld' quantization condition is that $2\pi I=nh$ where $n$ is an integer.}.

The \textit{angle} variable $\theta$ specifies where the phase space point is along the orbit, as illustrated by the simple harmonic oscillator case depicted in Fig.~\ref{cee}. ($\theta$, $I$) can be obtained, at least in certain cases, for a given Hamiltonian via a canonical transformation from ($q$,$p$)~\cite{Arnold}. More specifically, the Liouville-Arnold theorem says that an action-angle coordinate system exists for all completely integrable systems~\cite{Arnold}. Harmonic oscillators are prominent examples of completely integrable systems.
\begin{figure}[h]
  \centering
  \includegraphics[width=8cm,trim=120 60 90 50,clip]{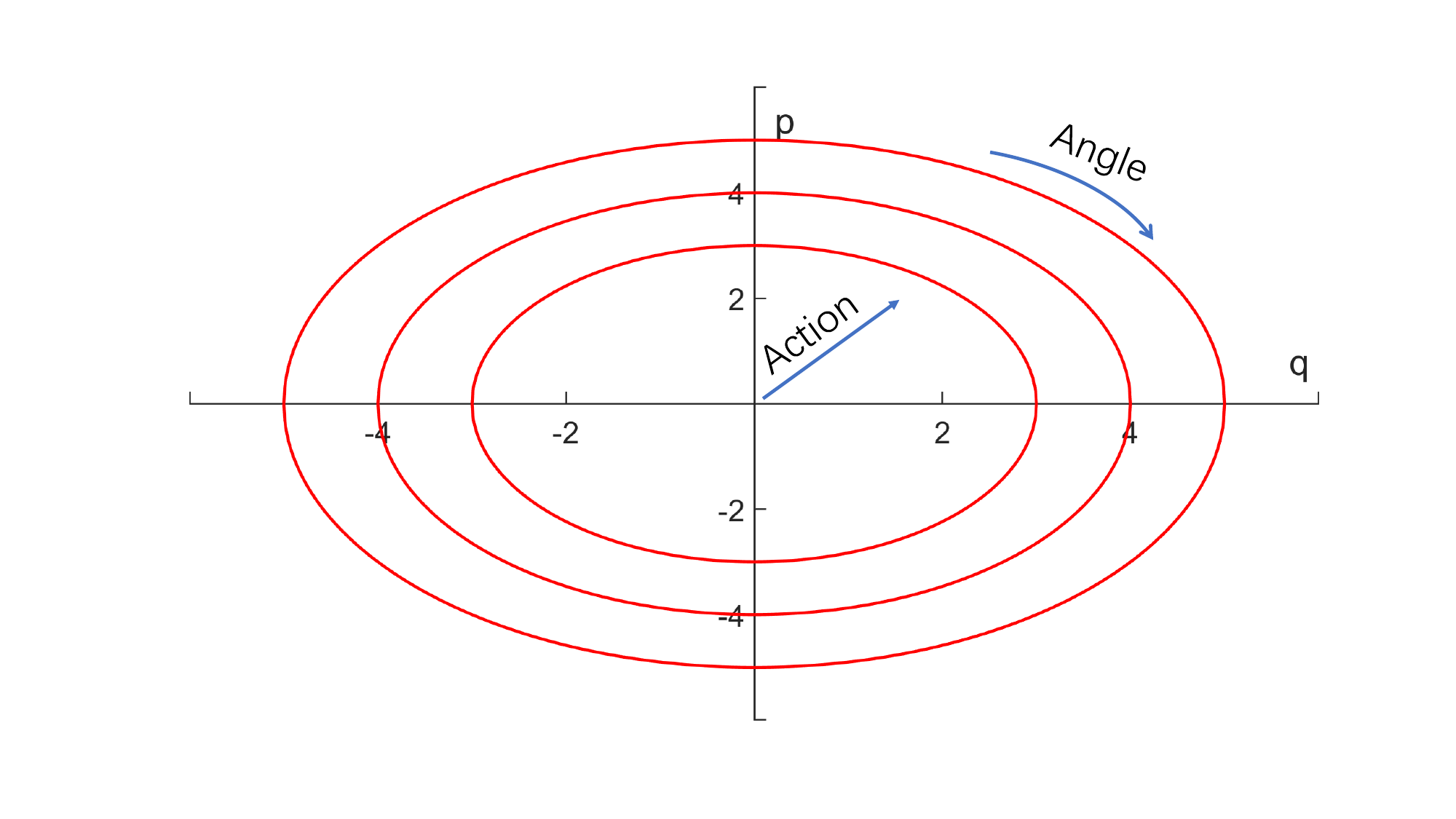}
  \caption{The action contour lines of a harmonic oscillator whose {$H=k(q^2+p^2)$. Both $q$ and $p$ are plotted in units of $\hbar^{\frac{1}{2}}$. Their actions from outside to inside are $\frac{25\hbar}{2\pi},\frac{16\hbar}{2\pi}, \frac{9\hbar}{2\pi}$}. The action contour lines correspond to orbits in phase space. The angle alone changes during the evolution. }\label{cee}
\end{figure}

The canonical equation for action-angle coordinates is
\begin{equation}\label{AAEOM}
\pdv{I}{t}=-\pdv{H}{\theta}=0,~~~~\pdv{\theta}{t}=\pdv{H}{I},
\end{equation}
where $H$ only depends on action. While the action is invariant for the time-independent Hamiltonian, the angle $\theta$ evolves at a constant speed: $\dot \theta(I,t)=\pdv{H}{I}$, indicating the phase of the periodic motion.

By the Liouville equation, Eq.~\eqref{Leq},
the evolution of the distribution $f$ under these coordinates is given by
\begin{equation}\pdv{f}{t}=-\pdv{f}{\theta}\pdv{H}{I}.\end{equation}

\subsection{Generalised  probabilistic theories }
 Generalised probabilistic theories (GPTs) {express} the idea that at the operational level, only statistics of measurement outcomes conditional on preparations and measurement procedures form the empirical basis of a theory, in contrast with indirect concepts like force. Therefore, GPTs are also called operational probabilistic theories. GPTs associate experiments on a system with real vectors, e.g.\, the probability vectors $\vec{f}$ corresponding to the individual measurements and outcomes~\cite{007GPT,014GPT,011InPoCo}, such as
\begin{equation}
\vec{f}=
\left[
\begin{array}{c}
\vdots\\
P(j|f, M_i)\\
P(j+1|f, M_i)\\
\vdots\\
P(k|f, M_{i+1})\\
P(k+1|f, M_{i+1})\\
\vdots
\end{array}
\right],
\end{equation}
where $P(j|f, M_i)$ represents the probability of the $j$-th outcome of the $i$-th measurement {on a} state $f$.

Mixed states are represented by the linear combinations of state vectors. For example,  $\vec f'=p\vec f_1+(1-p)\vec f_2, p\in[0, 1]$ represents a probabilistic mixture of state $\vec f_1$ and $\vec f_2$ with {weights} $p$ and $1-p$ respectively. 
{As a result} state spaces are always convex sets. Given a set of states, the states that cannot be obtained as a mixture of other states are, as in quantum and classical theories, called {\em pure} states.

The states {are assumed to }  transform linearly. The transformations are modelled as real matrices $T$ such that they respect probabilistic mixtures: $T(p\vec f_1+(1-p)\vec f_2)=pT\vec f_1+(1-p)T\vec f_2$. The set of allowed transformations must be such that allowed states are taken to allowed states by them. If the inverse matrix $T^{-1}$ exists and is an allowed transform one says the transformation $T$ is {\em reversible}.

The measurements are represented by the dual elements of state space in a certain sense. One can introduce linear operators $e_{M,j}$ so that $e_{M,j}(f)=P(j|f, M)$ gives the probability of the $j$-th outcome in measurement $M$. These operators are called \textit{effects} in GPTs.
For instance, the Born rule $P=\text{Tr}(\hat E_j \hat \rho)$ in quantum mechanics can be understood as an effect $\text{Tr}(\hat E_j...)$ applied on the state $\hat \rho$. Every measurement always ends with a result, which requires the set of effects corresponding to the measurement to be \textit{complete}. The completeness of effects means that for arbitrary states $f$ and measurement $M$, \begin{equation}\sum_j e_{M,j}(f)=1.\end{equation}

This work treats the phase space formalism as a continuous variable case of {a general} GPT formalism, following the connection between these two frameworks established in Ref~\cite{016TunNeg}. A generalised phase-space-like formalism was also employed recently in Ref.~\cite{PlavalaK22}.

{In the following we assume that in a given GPT valid states are normalised functions (or distributions) on $\gP$. Convex combinations of states are also states.} As part of specifying a given theory the state space will in general have further restrictions.

In GPTs, effects are linear functionals of states, such that probabilistic mixtures of states lead to corresponding probabilistic mixtures of measurement outcomes. In phase space, effects are described as
\be
P(i|f)=e_i(f)=\int h_i f dqdp,
\ee
where $h_i$ is some function of $q,p$.

The completeness condition $\sum_i e_i(f)= 1$ for arbitrary $f$ requires
\be
\sum_i h_i=1.
\ee
If only a finite region $\mathfrak{D}\subset\mathfrak{P}$ is of concern, the completeness condition becomes
\be
\sum_i h_i=1_\mathfrak{D}, \label{compleD}
\ee
where $1_\mathfrak{D}$ is a function that equals one when $(q,p)\in\mathfrak{D}$ and zero otherwise.

For continuous effects labelled by a continuous variable $\mu$, the probability of the outcome falling into a continuous interval $(\mu,\mu+d\mu)$ is $dP(\mu;d\mu|f)=\rho(\mu|f)d\mu$, where $\rho(\mu|f)$ is the probability density for the outcome $\mu$ given the state $f$.  The most general expression is
\be
\rho(\mu|f)=\int f(q,p)h_\mu (z)dqdp
\ee
For example,  classical (sharp) phase space localization has $\mu=(q,p)\in\mathfrak{P}$. The state is given by the Liouville density $f=\rho(q,p)$. The probability of being within the  volume $dq_0dp_0$ around the point ($q_0,p_0)$ in $\mathfrak{P}$ is $dP=\rho(q_0,p_0)dq_0dp_0$.  In this case, $h_{(q_0,p_0)}(q,p)=\delta(q-q_0)\delta(p-p_0)$.

GPTs also include the conditional update rule after measurements
\be
f{\xrightarrow{i} }g_{(f,i)},
\ee
of the measured state $f$ if the outcome $i$ was registered. For example, the sharp classical measurement with eh outcome $(q_0,p_0)$ leads to the update $\rho\to \delta(q-q_0)\delta(p-p_0)$. In quantum mechanics, von Neumann measurement collapses the wave function to the projector that describes the effect, while the most general state transformer is given by the Kraus matrices \cite{BGL:95,BLPY:16}. The state update rule after measurements will not be discussed in this paper.

\section{Generalised canonical coordinate symmetries with a unique inner product}\label{coordinate}

In this section, we demand certain symmetries on $\gP$, and derive an inner product for quasi-probability distributions on $\mathfrak{P}$. The inner product will provide a generalization of the Born rule (up to a constant which is determined in the subsequent section). The rule gives the operational meaning to {states} in terms of probabilities of measurement outcomes for given system preparations. The symmetries restrict the Born rule to a natural mathematical generalization of the classical and quantum cases.

We shall demand certain elementary symmetries both to narrow down the Born rule and the time evolution. For simplicity, we focus on the case of a two-dimensional phase space with the coordinates $(q,p)$. {At a minimum, valid states need to be normalised functions in state space. Convex combinations of valid states are also valid states. As part of specifying a given theory the state space will in general have further restrictions.}

We demand that there exists a coordinate system of generalised position and momentum, $(q,p)$, that satisfies the following `canonical' coordinate symmetries.
\begin{assume}[Canonical coordinate symmetries]\label{canon}
There exists such a coordinate system $(q,p)$  where the physical laws manifested by equations of motion and measurement are invariant under the following coordinate transformations:\\
 1. Translation: $(q,p,t)\mapsto(q+a,p+b,t{+c}),$ for any $a,b,c\in \mathbbm R$. {We represent its action on functions via $(\hat{T}_{a,b,c}f)(q,p,t)=f(q+a,p+b,t+c)$}.\\
2. Switch: $(q,p,t)\mapsto (Cp,q/C,-t)$, where $C$ is an arbitrary constant with units $[C]=[ {q}/{p}]$. \\
3. Time reversal: $(q,p,t)\mapsto (q,-p,-t)$. (equivalent to $(q,p,t)\mapsto (-q,p,-t)$ by switch.)
\end{assume}
{Invariance under spacetime translations and boosts are one of the basic symmetries of nature, and time reversal symmetry, while approximately correct in low-energy physics, is a useful computational tool \cite{H:96}. The switch symmetry is not usually explicitly presented, though exists in classical mechanics  (via canonical transformations \cite{Arnold,G:80}) and quantum mechanics (via corresponding unitary and antiunitary transformations \cite{H:96,BLPY:16}). It can be viewed as placing position and momentum on an equal footing. Together, the symmetries physically define a canonical coordinate system. These symmetries generate a group that includes other symmetries. For example, we will utilize parity symmetry implied by time reversal and switch symmetry in the Appendix \ref{symmetry}.

In a coordinate system obeying these canonical symmetries, we can get a unique inner product by introducing an additional natural restriction.  Namely, the inner product between two states goes to zero with increasing separation between the configurations they describe.
{
\begin{assume}[Local inner product]\label{local}
The inner product is local which means for two arbitrary quasi-probability distributions $f_1$ and $f_2$,
  \begin{equation} \lim_{a\rightarrow\infty}\expval{f_1,\hat{T}_{a,0,0}f_2}=0.\end{equation}
\end{assume}}

These postulates identify the inner product.
\begin{theorem} [Generalised inner product]
Consider an inner product of two  {arbitrary} {quasi-probability distributions} $f_1,f_2$ of an elementary two-dimensional system. If the inner product is local and the generalised coordinates obey the canonical symmetries, then
 \begin{equation}\expval{f_1,f_2}\propto\int f_1(q,p)f_2(q,p)dqdp.\end{equation}\label{innerproduct}
\end{theorem}
\begin{proof}
An inner product is a bilinear symmetric function of two states. For phase space distributions $f_1$ and $f_2$, a general form of a bilinear function is
\begin{equation}
\label{eq:GeneralInnerProd}
\int M(q,p,\Delta q,\Delta p)f_1(q,p)f_2(q+\Delta q,p+\Delta p)dqdpd\Delta qd \Delta p,
\end{equation}where $M$ is an arbitrary function. The symmetric condition on the  inner product $\expval{f_1,f_2}=\expval{f_2,f_1}$ further requires
\begin{equation}
\label{eq:a}
M(q,p,\Delta q,\Delta p)=M(q,p,-\Delta q,-\Delta p)
\end{equation}
 for arbitrary $a,b,c,d\in \mathbbm{R}.$

Translation symmetry requires $\expval{f_1(q,p),f_2(q,p)}=\expval{f_1(q+a,p+b),f_2(q+a,p+b)}$ such that
\begin{equation}\begin{array}{l}\label{eq:translationM}
\int M(q,p,\Delta q,\Delta p)f_1(q,p)f_2(q+\Delta q,p+\Delta p)d\Omega=\\
\int M(q,p,\Delta q,\Delta p)f_1(q+a,p+b)f_2(q+a+\Delta q,p+b+\Delta p)d\Omega,
\end{array}
\end{equation}where $d\Omega=dqdpd\Delta q d \Delta p$.
Eq. \eqref{eq:translationM} holds for arbitrary $f_1$, $f_2$, so
\begin{equation}
M(q,p,\Delta q,\Delta p)= M(q-a,p-b,\Delta q,\Delta p),
\end{equation}
for all $a,b\in \mathbbm{R}$. Therefore, $M$ only depends on the relative distance $\Delta q, \Delta p,$
\begin{equation}
\label{eq:b}
M(q,p,\Delta q,\Delta p)=M(\Delta q,\Delta p).
\end{equation}
Similarly, switch symmetry with dimensional constant $C$ requires $\expval{f_1(q,p),f_2(q,p)}=\expval{f_1( p/C,C q),f_2(p/C,C q)}$, which leads to
 \begin{equation}
 \label{eq:c}
M(\Delta q,\Delta p)=M(\Delta p/C,C\Delta q).
\end{equation}
Time reversal symmetry requires $\expval{f_1(q,p),f_2(q,p)}=\expval{f_1(q,-p),f_2(q,-p)}$, which leads to
 \begin{equation}
 \label{eq:d}
M(\Delta q,\Delta p)=M(\Delta q,-\Delta p).
\end{equation}
\begin{figure}
  \centering
  \includegraphics[width=8cm,trim=20 250 0 250]{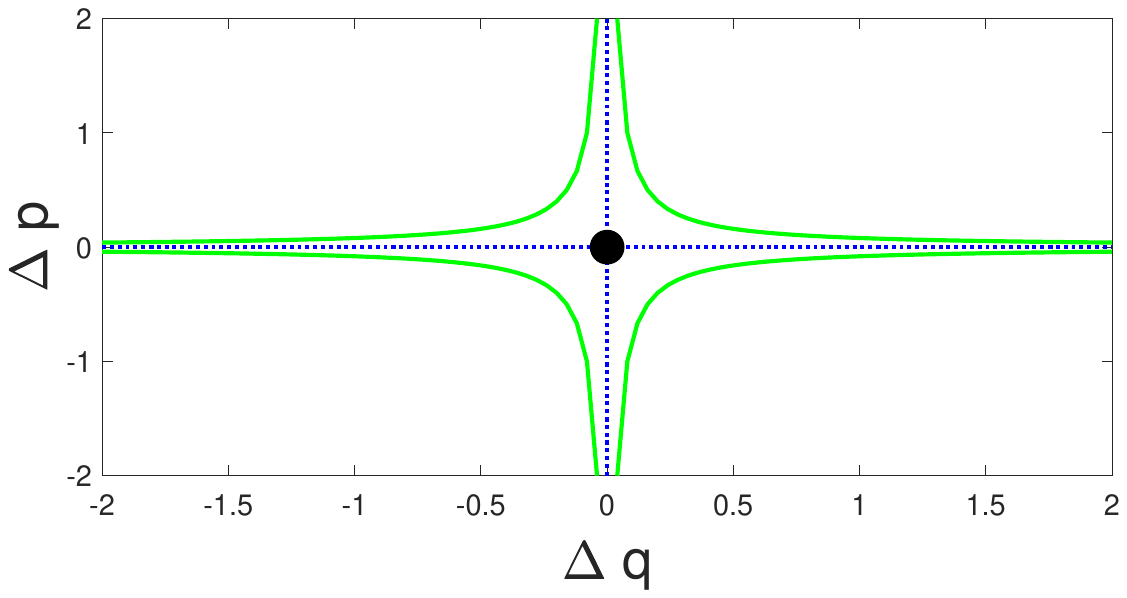}
  \caption{The contour plot of $M(\Delta q,\Delta p)$, a function that appears in the initially general form of the inner product as in Eq.~\eqref{eq:GeneralInnerProd}. {Both $q$ and $p$ are plotted in units of $\hbar^{\frac{1}{2}}$, but they can have different units as long as $[qp]=[\hbar]$}. The points in a pair of hyperbolas ($|\Delta q \Delta p|=c,c>0$) share the same value of $M$, and so does the $\Delta q,\Delta p$ axis (except the origin). The origin is an isolated point in the contour plot because it represents the inner product contributed by the point itself, unaffected by any symmetry operations. }\label{Inner}
\end{figure}

Equations~\eqref{eq:a}, \eqref{eq:b}, \eqref{eq:c}, \eqref{eq:d} imply that $M(q,p,\Delta q,\Delta p)$ is constant when $|\Delta p\Delta q|=c$ for arbitrary $c\geq 0$, except at the origin $(\Delta q=\Delta p=0)$. All these contour lines extend to infinity, as illustrated in Fig.~\ref{Inner}. Nevertheless, under Postulate~\ref{local}, $\langle f_1, f_2 \rangle=0$ for infinitely separated states, $M(\Delta q,\Delta p)$ must go to 0 when $\Delta q \rightarrow \infty$. This implies that $M(\Delta q,\Delta p)=0$ except for at the origin ($\Delta q=\Delta p=0$). Thus, $M(\Delta q, \Delta p)\propto \delta(\Delta q)\delta(\Delta p)$, and the inner product must have the form
\begin{equation}
\expval{f_1,f_2}\propto\int f_1(q,p)f_2(q,p)dqdp. \label{inprod0}
\end{equation}
\end{proof}

The inner product of two density matrices or Wigner functions gives the measurement probability under the standard Born rule. We  {present} a similar measurement formula in the generalised framework in Sec.~\ref{SD}, which determines the constant factor in the generalised Born rule and its relation to the Planck constant $h$.

{
\section{ {Generalised Planck constant of uncertainty: effect and state volume} } \label{SD}
 }
In this section, we introduce a property of effects that can be interpreted as the phase space volume they effectively occupy, which is called `effect volume'.
 Then, we introduce the generalised Born rule for a type of state-associated measurement, which is fundamental in quantum and classical theories. In this case, the effect volume can also be called state volume. The volume equals the Planck constant $h$ for quantum pure states and is zero (in a suitable limiting sense) for classical pure states.

{\subsection{Effects and their phase space volume}}
 \label{defvol}
Recall that the completeness of a measurement inside a region $\mathfrak{D}$ gives:
\be
\sum_i h_i=1_\mathfrak{D}.
\ee
Integrating  both sides results in
\be
\int \sum_i h_i dqdp=\int 1_\mathfrak{D}dqdp=V_\gD, \label{sumin}
\ee
where $V_\gD$ is the volume of the phase space region $\gD$.

We can factorise the $h_i$ as
\be h_i=c_ig_i, \label{decomposeeffect}
\ee
where $g_i$ is a normalized quasi-probability distribution, $\int g_idqdp=1$, and $c_i$ is some constant weight. Then interchanging summation and integration in Eq.~\eqref{sumin} results in
\be
\sum_i c_i=V_\gD. \label{sumV}
\ee

Eq.~\eqref{sumV}  identifies the sum of the weights $c_i$ with the phase space volume,
inspiring the following:
\begin{define}[{Effect volume}]
For the phase space representation of a discrete effect $e_i=(g_i,c_i)$ the \textit{effect volume} is defined as its weight,
\be
V_i\defeq c_i.
\ee

\end{define}
Hence the probability of the outcome $i$ can be written as
\be
e_i(f)=V_i\int f g_i dqdp. \label{probV}
\ee
Note that, despite its normalisation, $g_i$ does not necessarily represent a valid state in a GPT.

Aggregating different outcomes and thus combining different effects we count the total probability of the aggregate, arriving at a  \textit{coarse-grained measurement}. We can check that the effect volume is additive like `volume' in the sense that the coarse-grained effect volumes are the sum of aggregated effects' volumes.

\begin{theorem}[Coarse-grained measurement]\label{coarsegrain}
Consider a coarse-grained measurement effect, whose outcome probability is given by
\begin{equation} e_{C.G.}(f)=\sum_{i\in\mathbb{K}}e_i(f),
\end{equation}where $\mathbb{K}$ is a set of undistinguished results.
Then, its effect volume $V_{C.G.}$ is given by $V_{C.G.}=\sum_{i\in\mathbb{K}} V_i$.
\end{theorem}
\begin{proof}

\begin{equation} e_{C.G.}(f)=\int h_{C.G.}fdqdp =\sum_{i\in\mathbb{K}}V_i\int g_ifdqdp,
\end{equation}
where we defined $h_{C.G.}=\sum_{i\in\mathbb{K}}V_i g_i$. We factorise $h_{C.G.}$:
\be
h_{C.G.}= \sum_{j\in\mathbb{K}}V_j   \sum_{i\in\mathbb{K}}\frac{V_i g_i}{\sum_{j\in\mathbb{K}}V_j}.
\ee
 By Eq.~\eqref{decomposeeffect}, $h_{C.G.}=V_{C.G.}g_{C.G.}$ where $\int g_{C.G.} dqdp=1$.
Notice\be
\int \sum_{i\in\mathbb{K}}\frac{V_i g_i}{\sum_{j\in\mathbb{K}}V_j} dqdp=1,
\ee so the integrated function is $g_{C.G.}$ Consequently, $V_{C.G.}=\sum_{i\in\mathbb{K}} V_i$.
\end{proof}

{\subsection{Effect volume in state-dual measurements: state volume}
}

We now consider an important special case of effect volume, concerning effects that are associated with valid states. In this case, the effect volume can also be termed a state volume. We show that such a state volume is given by the inverse 2-norm of the distribution. We discuss the implications, including how the minimal state volume can be viewed as a generalisation of the Planck constant associated with uncertainty.

\begin{define}[{State-dual measurements}]
If all the effects $\{e_i\}$ of a measurement satisfy\\
(1) $e_i(f)\propto \expval{f,g_i}$, where $g_i$ is a valid state;\\
(2) $e_i(g_i)=1$.\\
 we call it a \textit{state-dual measurement}.\label{dualself}
\end{define}

Definition~\ref{dualself} leaves the measurement update rule of self-dual measurements general, e.g.\ it is not required that the measurement is repeatable (such that an iterated application always leads to the same result \cite{BGL:95,shortversion}).


State-dual measurements are particularly fundamental in \textit{self-dual} theories, which roughly means a one-to-one correspondence between states and effects \cite{014GPT,007GPT} up to normalisation. Both quantum and classical theories are self-dual. Another motivation for state-dual measurement comes from the dual role of Hamiltonian. While the Hamiltonian describes the evolution of states (property of states), it is also an observable (with effects), so it must bridge the states and effects. and we shall later use the generalised energy eigenstates to form a state-dual energy measurement just like the quantum case.

Physically, state-dual measurement effects can be understood as determining `Is the system in a state $g_i$?'. More specifically,
the definite outcome $e_i(g_i)=1$ implies that other outcomes must have zero probability: $e_i(g_j)=0$ when $i\neq j$. Consequently, states $g_i$ associated with a state-dual measurement must be orthogonal to each other:
\be
\expval{g_i,g_j}=0, \text{~when~} i\neq j. \label{ortho}
\ee



Definition~\ref{dualself} restricts the effect volume of state-dual effects. Firstly note that, since Theorem \ref{innerproduct} has derived a unique inner product, condition (1) in Definition~\ref{dualself} means $e_i(f)\propto \int fg_i dqdp$. The probability of an outcome $i$ is given by Eq.~\eqref{probV} but its volume $V_i$ is still undefined. Condition (2) identifies $V_i$.
\begin{theorem} [Volume of state-dual effects]\label{defvolume}
The volume of a state-dual measurement effect associated with state $g_i$ is given by
$V_i=1/\int g_i^2dqdp.$
\end{theorem}
\begin{proof} By condition (2) of Definition \ref{dualself},
\be e_i(g_i)= V_i\int g_i^2dqdp=1.\ee
Therefore,
\be V_i=1/\int g_i^2dqdp. \label{volume}\ee
\end{proof}

The $V_i$ of Eq.~\eqref{volume} only depends on the state $g_i$ therefore we similarly term $V_i$ the {\em state volume} of $g_i$. For normalized distributions,  the larger the 2-norm, the more peaked the function is. Thus the volume of a state, the inverse of the 2-norm, intuitively reflects the uncertainty of outcomes given that state. We shall later show how the state volume, minimised over states, becomes the Planck constant in the case of quantum theory.


The state volume is closely connected to the number of distinguishable states. Consider a complete state-dual measurement (or complete inside region $\gD\subset\gP$) with the associated set of states $\{g_i\}$. This measurement can distinguish these $\{g_i\}$ in a single shot. States and effects are symmetric in state-dual measurements, so we can apply the completeness condition of effects to states,
\be
\sum_i V_i g_i=1 \quad \big(\text{or} 1_\mathfrak{D}\big), \label{mms}
\ee
when $g_i$ are associated with a complete set of state-dual measurements (or complete inside region $\mathfrak{D}$). (The functions $1$ or $1_\mathfrak{D}$ on the right-hand side represent the (unnormalized) \textit{maximally mixed state}. When we know nothing about a system, it is in the maximally mixed state, which is the probabilistic mixture of all possible results by Eq. \eqref{mms}.)

Repeating the analysis of Sec.~\ref{defvol}, we know $\sum_i V_{g_i}=V_\gD$ when completeness is inside region $\mathfrak{D}$. If $V_{g_i}$ equals a constant $V_g$ for all $i$, then we can
count the total number $N_g$ of states inside $\{g_i\}$  by
\be
N_g=\frac{V_\gD}{V_g}.
\ee
Thus, the state volume determines how many orthogonal states can be stored in a finite phase space region.

These arguments can be extended to the case where the completeness condition is approximately defined $\big(\sum_i V_ig_i \simeq1_\gD\big)$.

{\subsection{Examples}}
In quantum mechanics, projective measurements are state-dual measurements. Their probabilities are given by
\be P(\phi| \psi)=e_\phi(\psi)=|\expval{\phi|\psi}|^2=h\int W_\phi W_\psi, \label{pWW}
\ee or the coarse-graining of the above outcomes.

The state/effect volume is associated with the purity of quantum states by
\begin{equation}
\mathrm{tr}(\hat\rho^2)=h\int W^2dqdp=h/V.
\end{equation}
Therefore, pure states have a minimal volume $h$. All the mixed states have a larger state volume. This also provides a reflection of the uncertainty principle in the Wigner function formalism.

The Eq.~\eqref{pWW} implements the so-called \textit{reciprocity law} of quantum mechanics \cite{Peres}, which states that for two pure states $\phi$ and $\psi$ the probability of observing outcome $\phi$ in a maximal test following preparation of state $\psi$ equals to the probability of observing outcome $\psi$ in a maximal test following a preparation of state $\phi$. If a GPT has two or more sets of state-dual effects, $e=\big(g_i,V_i\big)$, $e'=\big(g'_i,V'_i\big)$, \ldots that satisfy the reciprocity relation in the form
\be
P(i|j')=P(j'|i),
\ee
then all effects of all these measurements have the same volume. This follows from the application of Eq.~\eqref{probV} twice, reversing the roles of the state and the effect for the effects of two measurements.

Position and momentum eigenstates have state volumes different from `regular' pure quantum states consistent with the fact that they are actually outside the Hilbert space. They are unnormalizable delta functions. One may approximate them by Gaussians with finite width, in which case the state volume is still $h$.

 {The classical sharp phase space localization has $\mu=(q_0,p_0)\in\mathfrak{P}$. As we have seen in  Sec.~\ref{defvol}  the effect $(q_0,p_0)$ (with the  `uncertainty' $dq_0dp_0$) is represented by a normalized distribution  $g_{(q_0,p_0)}(q,p)=\delta(q-q_0)\delta(p-p_0)$ on $\gP$ with $c_{(q_0,p_0)}=1$.} As a result,
\be
dV_{(q_0,p_0)}=dq_0dp_0.
\ee

The coarse-grained version of such sharp effects provides another example. 
As a simple example consider the box mixed state:
\begin{equation}
 g_0(q,p)= \left\{ \begin{array}{ll}
 ( \epsilon\delta)^{-1},& q\in(q_0,q_0+\epsilon) \text{~and~} p\in(p_0,p_0+\delta)\\
 0,& \text{otherwise}
 \end{array} \right..
  \end{equation}
 The state volume that is occupied by each of these states is $V_g=\epsilon\delta\to 0$, and a domain $\gD$ contains
\be
N_g=V_\gD/(\epsilon\delta)\to \infty ,
\ee
orthogonal states.

\section{Generalised energy eigenstates: pure stationary states}\label{Pss}
The concept of eigenstate appeared with the foundations of quantum mechanics; it also appears in the quantum-like Koopman formalism of classical mechanics~\cite{931KooMec,Peres,PT:01,S:12,T:23}. However, here we will define the energy eigenstates from a slightly different perspective, which also works for generalised theories.
We define the generalised eigenstates via the dynamic features of states: the purest stationary states. The Noether theorem inspires this: the energy directly represents the time-invariant feature. 
These states are a cornerstone to put quantum and classical mechanics in the same framework.

\subsection{Definition}

{
\begin{define}[Stationary states]  Given a time evolution rule, \textit{stationary states} are states represented by time-independent phase space functions.
\end{define}}
Probabilistic mixtures of stationary states are, by inspection, also stationary, so there is a convex set of stationary states.
\begin{define}[Pure stationary states]
\textit{Pure stationary states} are states in the set of stationary states that cannot be represented as non-trivial probabilistic mixtures of other stationary states.
\end{define}
Note that pure {\em stationary states} are not necessarily pure states of the convex set of {\em all} allowed states of a GPT.

We will identify pure stationary states with the generalised energy eigenstates. We show in subsequent sections that they satisfy three natural desiderata:\\
1. Pure stationary states can be assigned sharp energy values, always giving the same value in energy measurement.\\
2. They describe the time evolution of the system.\\
3. They coincide with the standard quantum energy states in the case of quantum mechanics.\\

\subsection{Examples}
In the classical case orbits in phase space describe the time evolution, and uniform distributions over orbits are the only stationary states. More specifically, stationary states by definition obey $\pdv{f}{t}=0$. In the action-angle coordinates described in Sec. \ref{sec:actionangleintro}, \begin{equation} \pdv{f}{t}=\pdv{f}{\theta}\pdv{H}{I}.\end{equation}
 Consider the case of $\pdv{H}{I}\neq 0$ first. Then, a stationary state requires $\pdv{f}{\theta}=0$ for all $\theta$, so
\begin{equation}f(I,\theta)=f(I)=\int \rho(I_i)f_i dI_i,\end{equation}
where\begin{equation}\label{classic estate}
f_i(I)=\frac{1}{2\pi}\delta(I-I_i)
\end{equation}
are thus (normalized) pure states of the set of stationary states. {We see that these states are not pure in the set of all states}. As an illustrative example, the delta functions in Fig.~\ref{cee} are examples of pure stationary states in harmonic oscillators.
The state-dual measurement of $I$ is given by
\begin{equation}
P(I_0,dI_0|f)=p(I_0|f)dI_0=2\pi dI_0 \int f \frac{1}{2\pi} \delta(I-I_0) dqdp.
\end{equation}
Therefore, pure stationary states have infinitesimal volume,
\begin{equation}dV_{I_0}= {2\pi dI_0},\end{equation}
in agreement with the usual geometric interpretation of the action variables \cite{Arnold,AKN:06}.

Finally, consider the case when $\pdv{H}{I}=0$, for some values of {$I$}, such that any distribution over $\theta$ for those values of $I$ is stationary. Hence, the associated classical pure stationary states become pure states $\delta(I-I_0)\delta(\theta-\theta_0)$, where $I_0\in\{I'|\pdv{H}{I}|_{I'}=0\}$,$\theta_0\in[0,2\pi)$. (The state volume is $V= dq_0dp_0=dI_0d\theta_0$ by the same reasoning as above.)

Classical pure stationary states correspond to `eigen-wavefunctions' in the Koopman-von Neumann  formalism~\cite{931KooMec,003KooMec},  where the classical states are described by wave functions like quantum mechanics. Koopman mechanics describes classical states by probability amplitudes $\phi(q,p,t)$,  and the probability density $\rho(q,p)=|\phi(q,p)|^2$. The evolution of $\phi(q,p,t)$ is given by $i\pdv{\phi}{t}=\hat{L}\phi$, where $\hat L=-i\pdv{H}{p}\pdv{}{x}+i\pdv{H}{x}\pdv{}{p}$, named the Liouvillian operator is a generator of time translations, analogous to the quantum Hamiltonian operator. The classical Hamiltonian $H$ that enters the expression for  $\hat L$ is the energy observable. For example, (unnormalized) eigenfunctions of a free particle Liouvillian are, $\phi_{p_0,\lambda}=e^{i\lambda_q q}\delta(p-p_0)$, where $\lambda_q$ and $p_0$ can be arbitrary real numbers. However, the corresponding probability distributions containing $\delta^2(p-p_0)$ are ill-defined. Instead, pure stationary states are effectively $|\phi|$ up to normalization in this case.  
Moreover, a conceptual separation between the energy observable and the time translations generator puts the  Koopmanian formalism outside the GPTs we consider.

For quantum mechanics, the time evolution can be described by the commutator of the density matrix and the Hamiltonian. Therefore, stationary states must have diagonal density matrices in the energy basis. Pure states among stationary states are just energy eigenstates, including in cases of degeneracy. Hence, the pure stationary state is just another name for the energy eigenstate in the quantum case.

By inspection, both quantum and classical pure stationary states have sharp energy values. This provides a vehicle to define energy measurement, which associates measurements with states. Whether the pure stationary states can provide an energy measurement in generalised theories will be established in Sec. \ref{energy}.

The next section shows how pure stationary states relate to time evolution in generalised theories. The dynamical symmetry represented by pure stationary states is connected to energy as a conserved quantity. This is consistent with Noether's theorem, as will be explained later.

\section{Generalised equation of motion from Postulates}\label{EOM}
We will now derive the equation of motion in terms of generalised eigenstates.

\subsection{Evolution based on pure stationary states}
As an analogy of quantum energy eigenstates, we introduce the following non-trivial Postulates.
\begin{assume}[Evolution dependence]\label{PSSEES}
The time evolution of a state only linearly depends on the \textit{pure stationary states}, up to some dimensional factors $\mathcal{E}_i$ to keep the dimensions identical.
\begin{equation}
\label{eq:PSSEES}
\frac{\partial f}{\partial t}=G\left(f,\sum_i \mathcal{E}_i g_i \right)=\sum_i \mathcal{E}_i G\left(f, g_i \right), \end{equation} where $g_i$ is a set of pure stationary states and $\mathcal{E}_i$ are corresponding parameters and $G$ is some bi-linear functional.
\end{assume}
There are several reasons for Postulate~\ref{PSSEES}.
Firstly, stationary states are the only choice of states, under time translation symmetry, to depict the time evolution without external factors because they are the only states independent of time. Secondly, using {\em pure} stationary states gives the evolution as much freedom as possible. {If we used some mixtures of stationary states instead of pure stationary states, the evolution is equivalent to a specific linear combination of pure stationary states, which is just a special case of this postulate.} Finally, the pure stationary states only contain time-independent information, while $\pdv{f}{t}$ contains the dimension of time. Hence, we must have some parameters $\mathcal{E}_i$  containing the dimension of the inverse time  (which will later turn out to be proportional to the energy value).

\begin{assume} [Independence of stationary states]
The pure stationary states are independent in the sense that $G\left(g_i, g_j\right)=0$ holds for arbitrary $i,j$.\label{pssees}
\end{assume}
When you change the value of $\mathcal{E}_i$, in principle, the original pure stationary states may not be stationary anymore.
Postulate~\ref{pssees} aims to avoid such a complex situation. All the pure stationary states are stationary under other pure stationary states' impact so that their stationarity is independent of $\mathcal{E}_i$ (energy values). Both quantum and classical mechanics satisfy this postulate by inspection.


\subsection{Localized and non-localized evolution expressions}\label{localnonlocal}
Before deriving the equation of motion, we will introduce a new perspective to describe the time evolution, which will benefit later work. We will discuss the idea of the evolution of the phase space state distribution being localized in phase space.

Often time evolution is written as an {\textit{apparently localized expression}}. For example, {in} the continuity equation $ \partial_t f+\nabla\cdot (\vec vf)=0$ for conserved distributions evolution of $f$ at a point being determined locally depends on the velocity field $v$ at this point. In general, an apparently localized expression (of evolution) means the state update at one point only depends on quantities (derivatives) at this point. The Liouville equation \eqref{Leq} is an apparently localized expression in classical mechanics. On the other side, Eq.~\eqref{WFEOM} is an apparently local expression to describe quantum evolution.

In contrast, the Wigner transport equation (Eq.~\eqref{wignereq}) associates the evolution of one phase-space point evolution with non-adjacent Wigner function points. When the evolution equation at one point contains terms at other points, we call the equations {\textit{apparently non-localized expressions}} (of evolution). An important example is Eq.~\eqref{wignereq}. The original Eq.~\eqref{wignereq} works for the special case $H=p^2/2m+V(q)$. We may generalise it to arbitrary Hamiltonians:
\begin{equation}\label{GWignerequation}\centering
\begin{array}{l}
~~~\frac{\partial W}{\partial t}(q,p)=\int  W(q+l,p+j) J(q,p,l,j)djdl,\\
J(q,p,l,j)\\
=\frac{i}{\pi^2\hbar^3}\int[H(q-y,p+z)-H(q+y,p-z)] e^{-2\frac{i}{\hbar}(jy+lz)} dydz.
\end{array}
\end{equation}

Eq.~\eqref{GWignerequation} is equivalent to the quantum Liouville equation \eqref{WFEOM} (as shown in Appendix \ref{deltatrick}). 
From the Eq.~\eqref{GWignerequation} representation of time-evolution, we find the evolution at $(q,p)$ to depend on the distribution at $(q+l,p+j)$ (and thus everywhere). The probability conservation is guaranteed by $J(q,p,l,j)=-J(q+2l,p+2j,-l,-j)$, so the term $W(q+l,p+j)J(q,p,l,j)$ actually contributes to $\pdv{W(q,p)}{t}$ and $-\pdv{W(q+2l,p+2j)}{t}$. The distribution $W(q+l,p+j)$ in that sense plays the role of a porter, transferring other distributions from $(q+2l,p+2j)$ to $(q,p)$. This differs from the standard classical stochastic evolution wherein the probability current from A to B is always proportional to P(A).

A general apparently non-localized expression {(only assuming $\pdv{f}{t}$ is linear in $f$)} is as follows:
\begin{equation}
\label{eq:GWignerequationf}
\frac{\partial f}{\partial t}(q,p)=\int f(q+l,p+j)J(q,p,l,j)dldj,
\end{equation}
where the state update is described by a non-localized generator $J$.  In quantum and classical theories, $J$ depends on the Hamiltonian. We will apply the apparently non-localized expression in the following derivations.
By inspection, Eq.~\eqref{eq:GWignerequationf} only assumes that the evolution is linear in $f$, respecting the property of (quasi-)probability distributions.

Apparently localized and non-localized expressions themselves do not imply any physical difference, they are just two different languages to express the evolution rule. The apparently non-localized Eq.~\eqref{eq:GWignerequationf} can transform to an apparently localized expression like Eq. (\ref{WFEOM}) and vice versa (for well-behaved functions). Consider classical free particles as an example, $H=\frac{p^2}{2m}$, $\pdv{f}{t}=-\pdv{f}{q} \frac{p}{m}$. We can express the same evolution by $J(q,p,l,j)=\delta'(l)\delta(j)\frac{p}{m}$  in the non-localized expression ($\delta'(~)$ is the derivative of the delta function). (Appendix \ref{deltatrick} shows the derivation.)

We introduce the apparently non-localized expression because it is more convenient for describing {\textit{physically non-localized evolution}}. The physically non-localized evolution here denotes the appearance of infinite-order derivative \cite{948Local,950NoLoAc} in the apparently localized expression. Quantum mechanics has a physically non-localized evolution, which can be seen from the $\sin(\Lambda)$ term of the Eq. \eqref{WFEOM}.

On the other hand, if the evolution can be described by finite-order derivatives in the apparently localized expression, then we call it {\textit{physically localized evolution}}. Classical mechanics has physically localized evolution since the Liouville equation \eqref{Leq} is an apparently localized expression that only contains first-order derivatives.

When describing physically non-localized evolution by apparently localized expressions, we have to deal with infinite-order derivatives and their physical meaning is abstract.  However, beginning with an apparently non-localized expression like Eq.~\eqref{eq:GWignerequationf} can explicitly demonstrate the `jumping' in phase space, which has a clear physical picture. Since our generalised framework may contain physically non-localized evolution, we choose apparently non-localized expressions to derive the generalised equation of motion. We also find it eases the derivation process.

\subsection{Conditions to derive the equation of motion}
Before deriving the equation of motion, we will present several physical requirements and see how these restrict $J(q,p,l,j)$ in the equation of motion.

Postulate \ref{PSSEES} says the evolution linearly depends on pure stationary states, which is the core of our framework. We also have the canonical coordinate symmetries introduced in the Postulate \ref{canon}. We further require:

\begin{assume}[Inner product invariance]
The time derivative of inner products $\frac{\partial}{\partial t}\int f_1(t)f_2(t)dqdp=0$ for arbitrary states $f_1,f_2$ and time point $t$. \label{infcon}
\end{assume} The postulate has several consequences. First, {combined with Definition \ref{dualself}, it implies that applying a state-dual measurement on another state has a time-independent probability when both states evolve under a given evolution. For example, if initially you know a system has a 50\% probability of being state $f(t=0)$, then during the evolution, we expect the system to still have a 50\% probability of being $f(t)$. This is a physical motivation of Postulate \ref{infcon}.}
 Second, $\frac{\partial}{\partial t}\int f^2dqdp=0$ for any $f$, which means the state volume does not change. Third, the maximally mixed state is the state that maximizes the state volume, so it has to be stationary. {Thus it is a convex combination of pure stationary states and accordingly the pure stationary states satisfy the completeness condition.} Finally, any state $f$'s inner product with the maximally mixed state is invariant under time evolution, $\frac{\partial}{\partial t}\int f 1dqdp=0$, which means the total probability is conserved. {The inner product invariance is the only use of the inner product in the derivation of the equation of motion.}

\subsection{Restrictions to the non-localized generator J}
Next, we will restrict the non-localized generator $J(q,p,l,j)$  by the above requirements.

Postulate \ref{PSSEES} says the evolution linearly depends on pure stationary states, therefore $J$ linearly depends on pure stationary states $g$. It gives us
\begin{equation}J_g(q,p,l,j)=\int g(q+x,p+y) A(q,p,x,y,l,j)dxdy,  \end{equation} where $A$ is some unsettled function. What appears in the final equation of motion is $J_{\sum_i \mathcal{E}_ig_i}$, but here we only focus on the functional $J$. This is similar to Eq. \eqref{eq:GWignerequationf} for the same reason: this is the most general form of linear dependence. However, we step further and argue that:
\begin{lemma}[$J$'s dependence on $g$]
c\begin{equation}J_g(q,p,l,j)=\int g(q+x,p+y) A(x,y,l,j)dxdy.\label{Jg}\end{equation} $A$ is some unsettled function, and cannot depend on $q,p$ and {state $g$}.\end{lemma} The reason that $A$ cannot depend on $q,p$ is the following: The coordinates themselves are physically meaningless; all the dependence on $(q,p)$ means dependence on some state. We have already assumed that the dynamics are independent of time (Postulate \ref{canon}), and the only states independent of time are pure stationary states or linear combinations of them. If $A$ depends on pure stationary states, it breaks the linearity relation between $J$ and $g$. Therefore, there is no space for $A$ to depend on $q,p$\footnote{Whenever we utilize symmetries to restrict some function, we always implicitly assume the function only depends on the elements that we are aware of in this issue. Otherwise, arbitrary symmetry can be achieved by introducing a new degree of freedom, just like gauge fields.}. Eq.~\eqref{Jg} gives a general form of $J$ under Postulate \ref{PSSEES}. An equivalent statement is that if we translate all the pure stationary states, $g'(q,p)=g(q+\Delta q,p+\Delta p)$ without changing anything else, then $J_{g'}(q,p,l,j)=J_{g}(q+\Delta q,+\Delta p,l,j)$.

Then, to keep the inner product invariance of  Postulate~\ref{infcon},
\begin{widetext}
\begin{equation}\begin{array}{rl}
0=&\int\frac{\partial f_1}{\partial t}f_2+\frac{\partial f_2}{\partial t}f_1dqdp\\
=&\int f_1(q+l,p+j)J(q,p,l,j)f_2(q,p)+f_2(q+l,p+j)J(q,p,l,j)f_1(q,p)dqdpdjdl\\
=&\int f_1(q+l,p+j)J(q,p,l,j)f_2(q,p)+f_2(q-l,p-j)J(q,p,-l,-j)f_1(q,p)dqdpdjdl\\
=&\int f_1(q+l,p+j)J(q,p,l,j)f_2(q,p)+f_2(q,p)J(q+l,p+j,-l,-j)f_1(q+l,p+j)dqdpdjdl. \end{array}\end{equation}\end{widetext}(When the dependence on $g$ is not emphasised, we write $J(q,p,l,j)$ instead of $J_g(q,p,l,j)$.)
We require
\begin{lemma}[Requirement from inner product invariance]\label{innpro}
\begin{equation}J(q,p,l,j)=-J(q+l,p+j,-j,-l).\end{equation}
\end{lemma}
 This is consistent with the known fact that the generator of an orthogonal group, i.e., the group keeping the dot product invariant, is anti-symmetric.

Subsequently, we utilize the canonical coordinate symmetries introduced in Postulate \ref{canon}.
Due to translation symmetry, we can focus, without loss of generality, on the evolution of one point $(q=0,p=0)$  defining a functional: \begin{equation}J^0_g(l,j)=J_g(q=0,p=0,l,j).\end{equation}
For the equation of motion
\begin{equation}\pdv{f}{t} \text{}(0,0)=\int {J^0}_{g}(l,j) f(l,j)dldj,\end{equation} we apply a transformation to it,\footnote{We can write $(q,p,r)\mapsto(q',p',t')$, or equivalently $f(q,p,t),g(q,p,t)\mapsto f'(q,p,t),g'(q,p,t)$. We choose the latter one here.}
\begin{equation}\pdv{f'}{t} \text{}(0,0)=\int {{J^0}'}_{g'}(l,j) f'(l,j)dldj.\label{exofsym}\end{equation}
A symmetry means that the functional is invariant under a transformation: ${J^0}'=J^0$. The canonical symmetries will restrict $J$ in such a way.
The concrete derivation is in Appendix \ref{symmetry}. In the end, we find the canonical symmetries (Postulate \ref{canon}) require:
 \begin{lemma}[Requirement from canonical coordinate symmetries] \label{symlem}
\begin{equation}
J(q,p,l,j)=-J(q,p,-l,-j).
\end{equation}
\end{lemma}
 Lemma \ref{innpro}, \ref{symlem} (the requirements of $J$ from the inner product invariance and the canonical symmetries) are illustrated by Fig. \ref{sym}.
\begin{figure}
  \centering
  \includegraphics[width=8cm,trim=80 60 80 60]{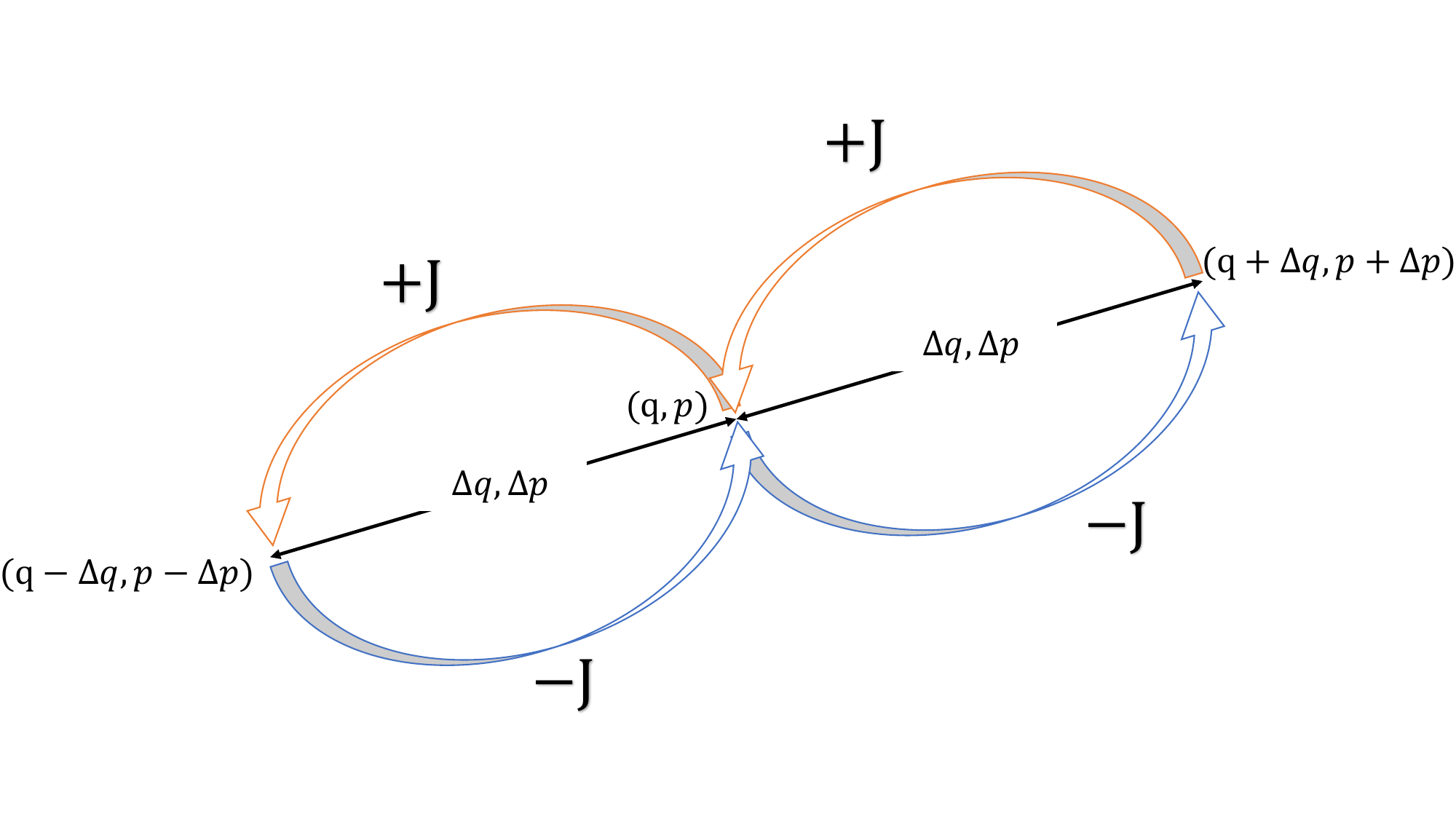}
  \caption{The figure illustrates the symmetry restrictions that Lemmas \ref{innpro} and \ref{symlem} impose on the non-localized generator $J$.  There are three phase space points: $(q,p)$, $(q+\Delta q,p+\Delta p)$, $(q-\Delta q,p-\Delta p)$. The restriction is that for any $q$,$p$, $\Delta q$ and $\Delta p$, $J(q,p,\Delta q,\Delta p)=-J(q,p,-\Delta q,-\Delta p)=-J(q+\Delta q,p+\Delta p,-\Delta q,-\Delta p).$}\label{sym}
\end{figure}

\subsection{Derivation of the equation of motion}
We already have all the ingredients on the table to derive the functional $J_g(q,p,l,j)$. By the end of this section, we will find that the evolution is similar to the quantum evolution but with extra degrees of freedom.

Fig.~\ref{sym} which combines Lemma \ref{innpro} and \ref{symlem} shows that
\begin{equation}\label{period} J(q,p,l,j)=J(q+l,p+j,l,j),\end{equation}
which means that $J$ is periodic in $q,p$ space.
Therefore, $J$ only has the components that satisfy $k_ql+k_pj=2\pi n~(n\in Z)$ in the frequency domain ($k_q,k_p$ are the angular frequency of $q,p$, respectively).

An observation will simplify the calculation: Eq. \eqref{Jg} is a convolution between $g$ and $A$. The convolution theorem gives:
\begin{equation}
J(q,p,l,j)=\text{Re}\int \tilde g(k_q,k_p) A''''(k_q,k_p,l,j)e^{i(k_qq+k_pp)}dk_qdk_p,
\end{equation} where $\tilde g$ is the Fourier transform of $g$. (All the $A$ with primes are unsettled functions, which help to absorb unimportant parameters.)

The periodic property leads to a $\sum_n\delta(k_ql+k_pj-2\pi n)$ term in the frequency domain. \small
\begin{equation}\begin{array}{l}\label{eq:fourierJ}
J(q,p,l,j)\\
=\text{Re}\int \tilde g(q_k,p_k) A'''(k_q,k_p,l,j)\sum_n \delta(k_ql+k_pj-2\pi n)e^{i(k_qq+k_pp)}dk_qdk_p\\
=\text{Re}\int g(q+y,p+z) \sum_n A''\left(\frac{2\pi n-k_pj}{l},k_p,l,j\right) e^{i(\frac{2\pi n-k_pj}{l}y+k_pz)}dk_pdydz\\
=\text{Re}\int g(q+y,p+z) \sum_n A'(n,k',l,j) e^{i\frac{2\pi ny}{l}} e^{-ik(jy-lz)} dk'dydz,
\end{array}\end{equation}
\normalsize
where we relabelled $\frac{k_p}{l}$ by $k'$. The term $e^{i\frac{2\pi ny}{l}}$ is not well-defined when $l=0$, but this term will vanish later.

The anti-symmetric condition Lemma \ref{symlem} requires that $J$ is an odd function of $l,j$, so
\begin{equation}J_{g}= \text{Im} \int g_i(q+y,p+z) \sum_n A'(n,k',l,j)  e^{i\frac{2\pi ny}{l}} e^{-ik'(jy-lz)} dk'dydz,\end{equation}and $A'(n,k',l,j)=A'^*(n,k',-l,-j)$.

One more requirement is that $J_{g}$ must keep $g$ itself stationary (Postulate \ref{pssees}), i.e.\ that
\begin{equation}\label{eq:kprimeconstraint}
\begin{array}{l}
~~~~\int g(q+l,p+j)J_g(q,p,l,j)dldj \\
=\text{Im}\int g(q+l,p+j) g(q+y,p+z)\\ ~~~~~~~~~~~~~~~~~~~~~~~~~~~\sum_n A'(n,k',l,j)  e^{i\frac{2\pi ny}{l}} e^{-ik'(jy-lz)} dk'dydzdldj\\=0.
\end{array}
\end{equation}
Observe that the equation can be written in the matrix form: \begin{equation}g_{lj}M^{ljyz}g_{yz}=0,\end{equation} where $M^{ljyz}=\text{Im}\int\sum_n A'(n,k',l,j)  e^{i\frac{2\pi ny}{l}} e^{-ik'(jy-lz)}dk'$.
This relation holds for arbitrary vector $g$, the matrix $M$ turns every vector into an orthogonal vector. It means $M$ must be a generator of the orthogonal group, which is anti-symmetric, $M^{ljyz}=-M^{yzlj}$, swapping $yz$ with $lz$ changes its sign. Therefore, we require that $n$ can only equal zero and $A'(n,k',l,j)=A'(k')$.
Now the form of $J$ is
\begin{equation}
J_g(q,p,l,j)=\text{Im}\int g(q+y,p+z) A'(k')  e^{-ik'(jy-lz)} dk'dydzdldj.
\end{equation}

The equation is already very similar to the generalised Wigner equation Eq.~\eqref{GWignerequation} (taking the imaginary part of the whole equation is equivalent to taking the odd part of the integrand.)
To harmonize the notation with the standard formulation quantum equation, we relabel $2/k'=k$ and replace $A'(k')$ by $K(k)$ such that $K(k)dk=A'(k')dk'$.
\begin{theorem}[Generalised equation of motion]
Given that (i) the evolution depends linearly on the generalised energy eigenstates (Postulates \ref{PSSEES} and \ref{pssees}), (ii) inner product invariance (Postulate \ref{infcon}), and (iii) the symmetries of phase space (Postulate \ref{canon}), the equation of motion has the following form:

\begin{align}
&\pdv{f}{t}\notag\\
&=\sum_i \mathcal{E}_i\mathrm{Im}\Big\{\int f(q+l,p+j) g_i(q+y,p+z)  K(k) e^{-i\frac{2(jy-lz)}{ k}}d\Omega\Big\}\notag\\
&=-i\frac{\pi^2}{2}\sum_i \mathcal{E}_i \int k^2 K(k) \mathrm{Wigner}_{ k}\{[\hat f_{ k}, \hat {(g_i)}_{ k}] \} d k,\label{Commutator EOM} \end{align}
where $d\Omega=dkdydzdldj$, the $\mathcal{E}_i$ are temporary parameters mentioned in Postulate \ref{PSSEES}, $K(k)$ is a theory-specific distribution, $\text{Wigner}_k\{~\}$ represents the $\hbar=k$ Wigner transform, and $\hat f_k, \hat {(g_i)}_k$ are operators corresponding to $f,g$ under the $\hbar=k$ Weyl transform Eq.\eqref{Weyltrans} (their units are different from density matrices). (Appendix \ref{suan} derives the operator form of the equation of motion.)
\end{theorem}
In the second line of Eq. \eqref{Commutator EOM}, we use the Wigner-Weyl transform (Eq.~\eqref{Wignertrans}, \eqref{Weyltrans}) to highlight the similarity to the common quantum expression in terms of commutator $\pdv{\hat \rho}{t}=\frac{i}{\hbar}[\hat \rho,H]$ for which there is no integral over $k$.

A key generalization relative to the quantum and classical cases is the extra factor $K(k)$. In the quantum case, $K(k)=\delta(k-\hbar)$, and in the classical case $K(k)\rightarrow \delta(k)$. A non-trivial $K(k)$ can thus be interpreted as a linear combination of commutation relations instead of a single one in quantum mechanics. In quantum mechanics, this commutator gives
 $\hbar$ in quantum mechanics, so one may interpret this loosely as a linear combination of different $\hbar$s. To understand the qualitative meaning of $k$, notice that $k$ only appears in $e^{-i\frac{2(jy-lz)}{ k}}$ in Eq.~\eqref{Commutator EOM}. Namely, $k$ is proportional to the jumping distance $(l,j)$ in phase space, so $K(k)$ represents the jumping ability during evolution; hence we call it the \textit{non-localized dynamics kernel}.

\subsection{Examples}
Now, we will discuss the dynamics corresponding to different non-localized kernels $K(k)$.
For the case that $ K(k)$ is a single delta function $ K( k)=\delta( k-\kappa)$, the second line of Eq. (\ref{Commutator EOM}) is just a commutator in phase space, and the first line is, up to the appearance of the Hamiltonian, the quantum expression of Eq.~ (\ref{GWignerequation}) with $\hbar =\kappa$\begin{equation}
\begin{array}{l}
~~~~~~~~~~~~~\frac{\partial f}{\partial t}=\int  f(q+l,p+j) J(q,p,l,j)djdl,\\
J(q,p,l,j)=\int\sum_i\mathcal{E}_ig_i(q-y,p+z) e^{-2\frac{i}{\kappa}(jy+lz)} dydz.
\end{array}
\end{equation} (The Fourier transform of the odd part functions is equivalent to the imaginary part of the Fourier transform.) The Hamiltonian will be defined from $\mathcal{E}_i$, $g_i$ later.

 When $\kappa\rightarrow 0$, like in path integrals, the phase oscillates so fast that only the first-order derivatives of $f(q,p)$ and $g_i(q,p)$ contribute to the imaginary part of the integral. When $K(k)\rightarrow\delta(k)$ the equation returns to the classical case with physically localized evolution (Eq. \eqref{Leq}).
 \begin{equation}
 \pdv{f}{t}\propto \big\{\sum_i \mathcal{E}_ig_i,f\big\}.
 \end{equation}
However, even with the same equation of motion, quantum/classical mechanics are not the only possible theories because the state space could differ. In Sec. \ref{why}, we will introduce information-restricted quantum mechanics as an example.

The structure of stationary states can determine the value of $\kappa$. The Weyl transform under $\kappa$  can simultaneously diagonalize all the stationary states. However, $\kappa$ could differ from the state volume of stationary states divided by $2\pi$, like in information-restricted quantum mechanics. It can neither be understood as the `theoretically' (all the states without negative probability are allowed) minimal volume of stationary states. \footnote{Actually, $2\pi \kappa$ is indeed the theoretically minimal average volume of stationary states for finite-dimensional cases, but we cannot apply this statement to phase space.} Take quantum harmonic oscillators as an example. The evolution of the Wigner function is completely classical and independent of $\hbar$. It means that the theoretically minimal stationary state can be arbitrarily small.

Another unexpected possibility is that $K(k)$ is not a delta function.\footnote{\cite{022ToogeneralEOM} has proposed a similar but weaker (less restricted) equation independently, called the generalised Moyal bracket by authors.} A distribution of hybrid commutators describes the evolution. In such case, the Wigner function's associative Moyal product (star product) $A\star B= A e^{i \frac{\hbar}{2} \Lambda}B$, is replaced by a non-associative hybrid Moyal product $A\star_H B=\int A K(k) e^{i \frac{k}{2} \Lambda}Bdk$ (See Appendix \ref{associate} for the proof). Consequently, {when we do a Weyl transform to rewrite everything in the operator form, the operator product is non-associative}. Moreover, the dynamics do not allow for the decomposition of density matrices by wave functions. We can give a trivial example of where such dynamics work. In harmonic oscillators, the evolution for arbitrary $K(k)$ is completely classical; there can be an orthogonal and complete set of pure stationary states without negative probability.

It is interesting but less rigorous to apply Eq. (\ref{Commutator EOM}) directly to the $I-\theta$ coordinates (which do not rigorously satisfy our canonical symmetries). This equation of motion implies that if one canonical coordinate $\theta$ is periodic, then jumping in the $I$ direction has to be {\em discrete}. This strongly suggests that $I$ itself is discrete, consistent with the result of the discrete action-angle Wigner function in~\cite{000QuAcAn} and the Sommerfeld quantization condition that $2\pi I=nh$ where $n$ is an integer.

\section{Generalised Hamiltonian} \label{energy}
After settling down the equation of motion, we would like to reconstruct the Hamiltonian. We will then have Hamiltonian mechanics that can apply to quantum and classical theories and, in principle, more. Besides the generator of time evolution, the Hamiltonian also represents an observable, which contains the {\em measurement effects} and the {\em energy value}. We discuss each concept in turn before defining the Hamiltonian.

\subsection{Energy measurement effects}
Energy eigenstates provide the energy measurement effects in quantum mechanics. Generalizing this idea, we wish to define the measurement effects by the generalised energy eigenstates.  However, pure stationary states do not always, for any GPT, provide a set of effects for state-dual measurement. Sec. \ref{why} will give an example: Spekkens' toy model \cite{007SpeToy}.
Therefore, we have to postulate:
{\begin{assume}[Existence of energy measurement]There exists a state-dual measurement whose effects all correspond to pure stationary states. \label{measure}
\end{assume}
The postulate guarantees that pure stationary states can construct a complete and orthogonal set and the corresponding state-dual measurement exists.} The energy measurement constructed from pure stationary states provides all the information we can learn that is irrelevant to the time point and induces a conserved observable. The time translation symmetry, along with the inner product invariance, ensures that there always exists stationary states. Postulate \ref{measure} assures that there is always a state-dual measurement corresponding to pure stationary states. When stationary states give the set of measurement effects and any state's inner product with these stationary states is invariant (Postulate \ref{infcon}), then this set of effects always gives a conserved observable. Therefore, the energy measurement here agrees with the definition given by the Noether theorem: it is the conserved observable promised by the time translation symmetry.

\subsection{Energy value}
Next, we will determine the value of energy. Energy is usually associated with the time scale in the case of natural units ($\hbar=c=1$). This is also true for the $\mathcal{E}_i$, which appears in Postulate \ref{PSSEES} and the equation of motion \eqref{Commutator EOM} representing the speed of evolution. We, therefore, anticipate that the energy associated with stationary state $g_i$, which we shall call $E_i$, is proportional to $\mathcal{E}_i$.

The probabilities of energy eigenstates are always invariant by inner product invariance, so any time-independent function of $\mathcal{E}_i$ is conserved. To further constrain the definition of the energy value $E_i$, we demand that it is an extensive quantity (it adds under an independent composition of systems). To construct an extensive quantity,
consider a composite system with two independent subsystems (labelled by $1,2$) without any interaction and correlation, which means
\begin{equation}f_{12}=f_1f_2,~g_{12ij}=g_{1i}g_{2j},\end{equation}
and
\begin{equation}\label{eq:subsysevol}
\pdv{f_{12}}{t}= \pdv{f_1}{t}f_2+\pdv{f_2}{t}f_1.\end{equation}
We expect the energy in the composite system's equation of motion to be the sum of the subsystems' energy. Substituting the equation of motion Eq.~(\ref{Commutator EOM}) into Eq.~(\ref{eq:subsysevol}) gives
 \begin{widetext}
\begin{equation}\label{eq:f12decomposed}
\pdv{f_{12}}{t}=\sum_i \mathcal{E}_{1i}\text{Im}\int f_2(q_2,p_2)f_1(q_1+l_1,p_1+j_1) g_{1i}(q_1+y_1,p_1+z_1)K_1(k_1)e^{i2(j_1y_1-l_1z_1)/k}dy_1dz_1dl_1dj_1+1\leftrightarrows 2,\end{equation}
where the $1\leftrightarrows 2$ term is the same as the first term on the RHS except that indices 1 and 2 are interchanged. We next introduce the following identity: \begin{equation}\label{eq:fidentity}
f(q,p)=\int f(q+l,p+j)\delta(l)\delta(j)dldj=1/(\pi^2\bar {k^2})\int f(q+l,p+j)\sum_i V_ig_i K(k)e^{i2(jy-lz)/k}dldjdydzdk,\end{equation}
where $\bar {k^2}=\int k^2 K(k)dk$, and we have used $\delta(x)=\frac{1}{2\pi}\int e^{ikx}dk$, $\sum_i V_ig_i=1$. Replacing $f_2$ in Eq.~(\ref{eq:f12decomposed}) by Eqs.~(\ref{eq:fidentity}) gives
\begin{equation*}\pdv{f_{12}}{t}=\sum_{ij} \frac{ \mathcal{E}_{1i}}{\bar{k^2}_2\pi^2}V_{2j}\text{Im}\int f_2f_1 g_{1i}g_{2i}K_1(k_1)K_2(k_2)e^{i2(j_1y_1-l_1z_1)/k_1+(j_2y_2-l_2z_2)/k_2}(dydzdldj)_{1,2}+1\leftrightarrows 2\end{equation*}
\begin{equation}=\sum_{ij} \left(\frac{\mathcal{E}_{1i}}{\bar{k^2}_2\pi^2}(V_{2j})+ \frac{\mathcal{E}_{2i}}{\bar{k^2}_1\pi^2}(V_{1j})\right)\text{Im}\int f_2f_1 g_{1i}g_{2i}K_1(k_1)K_2(k_2)e^{i2[(j_1y_1-l_1z_1)/k_1+(j_2y_2-l_2z_2)/k_2]}(dydzdldj)_{1,2},\label{composite1}\end{equation}
where all the $f$ depends on $(q+l,p+j)$ and all the $g$ depend on $(q+y,q+z)$ with the corresponding subscript. On the other hand, the equation of motion for the composite system should have a consistent form
\begin{equation}\pdv{f_{12}}{t}=\sum_{ij} \mathcal{E}_ {ij}^{composite} \text{Im}\int f_1f_2 g_{1i}g_{2j}K_1(k_1)K_2(k_2)e^{i2[(j_1y_1-l_1z_1)/k_1+(j_2y_2-l_2z_2)/k_2]}(dydzdldj)_{1,2}.\label{composite2}\end{equation}\end{widetext}

Comparing Eq. (\ref{composite1}), (\ref{composite2}), we find that if $\mathcal{E}_i\propto \frac{E_iV_i}{\bar{k^2}}$ (we omit the subscript for subsystems), then ${\mathcal{E}_ {ij}^{composite}\propto (E_{1i}+E_{2j})\frac{V_{1i}V_{2j}}{\bar{k^2}_1\bar{k^2}_2}}$, i.e., $E$ is an additive quantity.

However, $\frac{\mathcal{E}_i\bar{k^2}}{V_i}$ has the unit (not natural units) $[\frac{1}{t}]$ ($[\mathcal{E}_i]=[\frac{1}{qpt}],[\bar{k^2}]=[q^2p^2],[V]=[qp]$) instead of the conventional quantum and classical unit for energy of $[\frac{qp}{t}]$. Of course, we may multiply an arbitrary constant with the unit $[qp]$ without breaking conservation. However, there may not exist a simple choice of such parameters in generalised theories because $V_i$ can be different for different states and $K(k)$ can be different for different subsystems; there is no unique constant like $\hbar$ in our framework.

We need an extra assumption to construct energy with $[\frac{\hbar}{t}]$, for example: \textit{the pure states share the same state volume $V_p$ in a theory.}\footnote{When considering the interaction between quantum and classical systems, people tried to assign different Planck constants to different subsystems, but this has been argued to be impossible~\cite{999UniPla,004UniPla}. Nevertheless, there is also work pointing out special cases that avoid the prohibition~\cite{020GenUnc}. Our model allows a non-associative algebra which is seldom considered in the previous works, so we cannot directly apply the above results.}
This assumption is plausible though {\em a priori} not necessary for a theory to be self-consistent. Otherwise, we can also accept the generalised energy with dimension $[\frac{1}{t}]$.

To restore the classical and quantum energy value, we define the energy of the i-th pure stationary state $E_i$ to be related to $\mathcal{E}_i$ via
\begin{equation}E_i=\frac{\pi \mathcal{E}_i V_p \bar{k^2}}{4V_i},\label{eq:evalue}\end{equation}where $\bar {k_n^2}=\int k^2 K(k)dk$ as above.

In the quantum case,   $\bar{k^2}=\hbar^2$, $V_p=V_i=2\pi\hbar$. Thus, the relation between $\mathcal{E}_i$ and $E_i$ is :
$E_i=\frac{\pi \hbar^2 \mathcal{E}_i}{4}.$ In the classical case, the energy eigenstates are continuous and the label $i$ can naturally be replaced by the action $I$ of the phase space orbit. The energy $E(I)$ of the given eigenstate also obeys the Eq.~\eqref{eq:evalue} relation with the corresponding $\mathcal{E}(I)$. The $\mathcal{E}(I)$, which do not have a clear meaning in classical mechanics, usually diverge. Otherwise, the corresponding $E(I)$ equals zero, which does not contribute to the evolution.
We conclude that the definition of energy eigenvalue $E_i$ from quantum theory is consistent with the definition of classical mechanical energy $E(I)$ within this framework. They are both special cases of Eq.~\eqref{eq:evalue}.

\subsection{Definition of Hamiltonian}\label{Hamiltonian}
Finally, we can define a function in phase space corresponding to a generalised Hamiltonian.
\begin{define}[Hamiltonian] \label{Hami}
The Hamiltonian is a phase space function $H(q,p)$. {For a specific system it is given by:}
\begin{equation}H(q,p):=\sum_i E_i V_{g_i}{g_i},\end{equation}
where $g_i$ are pure stationary states, with state volume $V_{g_i}$, and $E_i$ are energy values corresponding to $g_i$.\footnote{Why is the quantum Hamiltonian only the combination of eigenvalues and eigenstates, but here we have an extra $V_{g_i}$? Recall that the Wigner function is the Wigner transformation of density matrices with an extra factor, $W(q,p)=\frac{1}{h}\text{Wigner}\{\hat \rho\}$. We always have an extra factor in phase space.} $V_ig_i$ is dimensionless so $H$ has the dimension of $E_i$ as expected.
\end{define}
Applying the Hamiltonian definition to the equation of motion Eq.~\eqref{Commutator EOM} gives the following.
\begin{theorem}[Equation of motion in terms of Hamiltonian]
The time evolution of a state $f(q,p)$ can be represented by
\begin{equation} \begin{array}{l}
\frac{\partial f}{\partial t}\\ =\left(\!\!\frac{4}{\pi V_p\bar {k^2}}\!\!\right)\mathrm{Im} \!  \int \!\! f(q+l,p+j) H(q+y,p+z)K(k)e^{i2(jy-lz)/k}d\Omega,\end{array}\label{geom1}\end{equation}where $d\Omega=dydzdldjdk$ and the non-localized dynamics kernel $K(k)$ is a theory-specific function specifying the jumping in phase space, as we discussed below Eq.~\eqref{Commutator EOM}. Again $\bar {k^2}=\int k^2 K(k)dk$.
An equivalent expression is
\begin{equation}\frac{\partial f}{\partial t}=\frac{4\pi}{V_p \bar {k^2}}\int K(k)k^2 f\sin(\frac{\Lambda k}{2})H dk.\label{geom2}
\end{equation}
\end{theorem}
We can compute the expectation value of energy by $\expval{E}=\sum_i P(i|f)E_i$, where $P(i|f)$ is the probability of  measuring the pure stationary state effect $g_i$ given that the system is in the state $f$.  By the Definition \ref{dualself} of state-dual measurements, $P(i|f)=\int V_{g_i}g_if dqdp$. Combining that with the definition of  the Hamiltonian (Definition \ref{Hami}), we have:
\begin{theorem}[Energy measurement in terms of Hamiltonian]
The expectation value of energy for state $f$ is given by
\begin{equation}\expval{E}= \int H fdqdp. \end{equation}
\end{theorem}
\subsection{Examples}
In quantum mechanics, the $K(k)$ is a single delta function $K(k)=\delta(\hbar-k)$, and $V_p=V_i=2\pi h$. The Hamiltonian is given by $H=\sum_i h E_i g_i$, where $g_i$ are the Wigner functions of energy eigenstates with eigenvalues $E_i$.

Eq. (\ref{geom1}) becomes
\begin{equation}\frac{\partial f}{\partial t}=\left(\frac{2}{\pi^2\hbar^3}\right)\text{Im}\int f(q+l,p+j) H(q+y,p+z)e^{i2(jy-lz)/\hbar}d\Omega.\label{quaneom1}
\end{equation}

This is equivalent to the original quantum equation of motion,  Eq.~\eqref{GWignerequation}. (Taking the imaginary part is equivalent to taking the odd part inside the Fourier transform.)

 Similarly, in the quantum case, Eq.~\eqref{geom2} becomes
 \begin{equation}
 \pdv{f}{t}=\frac{2}{\hbar}f\sin(\frac{\hbar}{ 2}\Lambda)H,\label{quaneom2}
 \end{equation}  which is exactly the quantum Eq. (\ref{WFEOM}).

In classical mechanics, there are infinite pure stationary states. $H(q,p)=\int E_i (\frac{1}{2\pi}\delta(I(q,p)-I_i)) 2\pi dI_i$, where $(\frac{1}{2\pi}\delta(I(q,p)-I_i))$ are the normalized pure stationary states with state volume $2\pi dI_i$. $K(k)=\delta(k)$ in the classical limit. We can take the limit of quantum equation \eqref{quaneom2},
\begin{equation}
\pdv{f}{t}=\lim_{\hbar\rightarrow0}\frac{2}{\hbar}f\sin(\frac{\hbar}{ 2}\Lambda)H,
\end{equation} it is obvious that only the $\frac{\hbar}{ 2}\Lambda$ term in the $\sin(\frac{\hbar}{ 2}\Lambda)$ has a non-zero contribution. We get
\begin{equation}\pdv{f}{t}=f\Lambda H,
\end{equation}which is the classical Liouville equation.

The limit of the apparently non-localized expression Eq. (\ref{quaneom1}) is more tricky. Like in path integrals, when the phase of the integrand is so fast, only the stagnation point will contribute. Since the apparently localized and non-localized expressions are equivalent, one finds $\pdv{f}{t}=f\Lambda H$.

\section{State volume and chaos}
One example where quantum and classical theories appear to differ qualitatively is \textit{chaos}. A common qualitative definition of (classical) chaos is as follows. If two close states' distance grows exponentially during evolution, we say the system is chaotic because a small error in the initial condition will destroy predictability \cite{014Chaos,016Chaos}. Classical chaos is known to be quite universal, e.g., in turbulence and weather systems. However, although quantum chaos' definition is still being debated \cite{989QuaCha,016Chaos,014Chaos}, it must be different from the classical case because the Schr\"{o}dinger equation is a linear function that cannot exponentially magnify the perturbation \cite{989QuaCha}. Here, we explain the different behaviors in the chaos aspect by our framework.
\begin{figure}
  \centering
  \includegraphics[width=8cm]{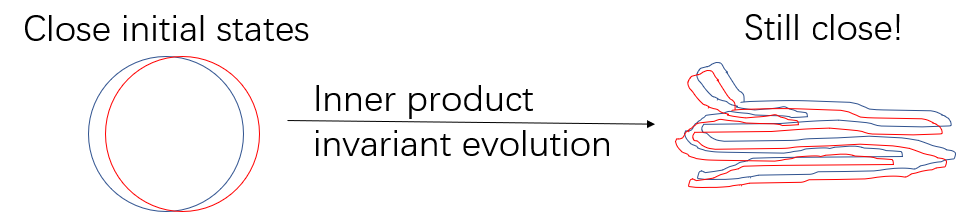}
  \caption{A sketch of how inner product invariance together with finite state volume restricts the evolution of perturbed states. When there is a finite state volume, there is a non-zero overlap of a state and a perturbed version of it. The demand that the inner product of the two states is invariant then restricts the motion significantly, counteracting the possibility of the final state being strongly sensitive to initial perturbations.}\label{chaos}
\end{figure}

The state volume explains the lack of sensitivity to initial state perturbation in quantum systems, as was noticed early~\cite{979QuaMap,981QuNoIn}. If states have finite state volume, the inner product invariance protects the perturbation from exponential growth (as depicted in Fig. \ref{chaos}). Similarly, if we consider a non-zero state volume state, like a Gaussian distribution, in a chaotic classical system, a perturbation to it will not lead to `chaos'. Because the perturbed distribution value changes negligibly at any point, you can easily estimate the evolution of the perturbed Gaussian distribution by the known evolution of the unperturbed state. Instead, simulating the evolution of a single Gaussian distribution requires heavy computation when it contains the chaotic evolution of infinite points. The pure quantum states themselves have non-zero state volume. Therefore we conclude that reversible quantum evolution is stable under small perturbations to states.\footnote{For a classically chaotic Hamiltonian, will it (always) be difficult to simulate the evolution of a single quantum state? If not, how does the quantum evolution ease the simulation? This question deserves further study.} On the other hand, when the perturbation is much larger compared with the state volume, the restriction given by inner product invariance is negligible.

We construct such a model: an initial state $f_I$ will evolve to $f_F$ under some inner-product-invariant evolution. If we perturb the initial state to $f'_I$, how much do we know about the final state $f'_F$? When the evolution is chaotic, we assume all the information about $f'_F$ is given by the inner product invariance, which says $f'_F$ should be as close to $f_F$ as $f'_I$ close to $f_I$ in terms of inner product (Fig. \ref{chaos}).

For simplicity, consider a Gaussian wave packet $f_I=e^{\frac{-q^2-p^2}{2 V_I/\pi}}$ and the perturbed one $f'_I=e^{\frac{-(q+\Delta x)^2-p^2}{2 V_I/\pi}}$. The $\Delta x$ is a random perturbation, we assume $\Delta x\in[-\sigma,\sigma]$, where $\sigma$ represents the strength of perturbation.  The perturbation $\Delta x$ may produce a set of states $f'_I$, whose inner products with $f_I$ are bounded by
\begin{equation}\langle f'_I,f_I\rangle \geq I(\frac{\sigma}{\sqrt{V_I}}).\label{innerconstrain}\end{equation} The function $I$ for Gaussian states can be seen in Fig.~\ref{Irelation}.
\begin{figure}
  \centering  \includegraphics[width=8cm,trim=0 290 0 250]{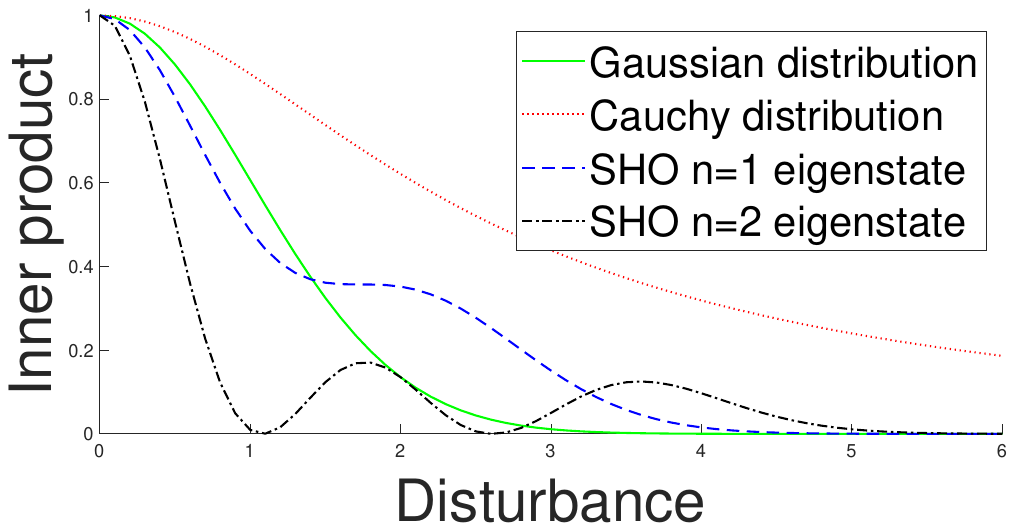}
  \caption{The function between the inner product with the original state and the disturbance $\Delta q$ for Gaussian distributions ($\exp(\frac{-q^2-p^2}{\hbar})$), Cauchy distribution ($\frac{\hbar}{q^2+p^2+\hbar}$), the n-th eigenstates of the simple harmonic oscillator (SHO, $H=k(q^2+p^2)$). {We normalized the inner products such that the inner products with themselves equal 1. The disturbance is plotted in units of $\hbar^{\frac{1}{2}}$}. The non-monotonicity and possible zero points are interesting phenomena for future studies.}\label{Irelation}
\end{figure}

The only information of the final state $f'_F$ comes from the inner product bound in Eq.~\eqref{innerconstrain}.  The information can be reflected by entropy $S=\sum_ip_i \log(\frac{1}{p_i})$, which can be estimated by $S\sim \log(N)$, where $N$ means the number of possible perturbed states. Since the evolution is a one-to-one map, we can count the number of possible initial states that satisfy Eq.~\eqref{innerconstrain}, which is equal to the number of possible final states that satisfy Eq.~\eqref{innerconstrain}.

Only an area $A\sim \sigma^2$ of states can satisfy Inequality \eqref{innerconstrain}, which can be counted as $N=\frac{A}{V_{p}}\sim\frac{\sigma}{V_{p}}$ pure states. Here we count the states by the pure state's volume $V_p$ and assume it is a constant for different pure states.  Notice the initial state $f_I$ might not be a pure state; therefore, we distinguish its state volume $V_I$ from $V_P$.

Now we can estimate the entropy\begin{equation}
S\sim \log(\frac{V_I}{V_p}+N)\sim \log(\frac{\sigma^2+V_I}{V_p}),
\end{equation} where we have added a $V_I$ term because when there is no perturbation at all, $\sigma=0$, we expect the entropy $S\sim \log(\frac{V_I}{V_p})$, which is the entropy of $f_I$. the above equation shows that, for a certain initial state with $V_I$, the ratio $\frac{\sigma^2}{V_p}$ determines the upper bound of the perturbed evolution's uncertainty.

Above, we considered the case of finite state volume. What is the situation in the classical case? First, as $V_p\rightarrow0$, the entropy blows up, which is common in the classical state count. However, there is another issue: originally, we estimate the area of possible states $A\sim \sigma^2$ under the inner product bound Eq.~\eqref{innerconstrain}, but Eq.~\eqref{innerconstrain} becomes trivial for classical pure states. All the inner products between different delta functions $\delta(q-q_0)\delta(p-p_0)$ are 0. Therefore, we cannot learn anything from inner product invariance, which means:
\begin{equation}
A\left\{\begin{array}{ll}
          \sim \sigma^2 & \text{for }~V_I>0, \\
          =\infty & \text{for }~V_I=0.
        \end{array} \right.
\end{equation}
This relationship does capture the sensitivity to the initial state in classical mechanics. However, the discontinuity of $A$ from a very small $V_I$ to $V_I=0$ is disturbing. We can modify our model to make it smooth. We introduce an uncertainty in the inner product $\Delta I$, which is a constant that might be caused by experimental limitations. The $\Delta I$ causes an extra term $\dv{\Delta x}{I}\Delta I$ in the perturbation. The relation between $I$ and $\Delta x$ is shown in Fig. \ref{Irelation}. When $V_I\rightarrow 0$, $\frac{\Delta x}{\sqrt{V_I}}$ and $\dv{\Delta x}{I}$ go to infinity. A small uncertainty in inner product $\Delta I$ will cause a large uncertainty area $A$ in the classical limit.

After introducing $\Delta I$, the modified model shows:
\begin{equation}
S\sim \log(\frac{V_I+\sigma^2+ (\Delta I \frac{d \Delta x}{d I})^2}{V_P}),
\end{equation} where $S$ is the upper bound entropy of the perturbed final state constrained by inner product variance, $V_I$ is the state volume of the initial state (can be mixed), $\sigma$ represents the strength of perturbation ($\Delta x\sim \sigma$), $\Delta I$ is an uncertainty in the inner product ($\pdv{I}{\Delta \emph{}x}$ is a function of $\Delta x$, but we roughly take it as a constant here), $V_P$ is state volume of pure state, which is used to count how many pure states are possible final states.  By this relation, we can see the state volume determines the restriction given by inner product invariance. {Although choosing this specific model to do this semi-quantitative analysis, we believe a similar tendency exists in the general case: the non-zero state volume restricts the system from evolving chaotically with the help of inner product invariance.} Quantum states have finite state volume, while classical pure states have zero state volume, this difference explains why classical systems can be sensitive to the initial state, but quantum systems are not.

\section{Two examples other than classical and quantum mechanics}\label{why}
\subsection{Case of information-restricted quantum mechanics}
We now create another example theory and analyze how our results apply there. In this theory, there is a limit to our knowledge about what pure quantum state a quantum system is in. This leads to a different relation between state volume and the evolution-related non-localized dynamics kernel, i.e., splitting between the generalised Planck constants associated with states and evolution, respectively. This theory can be called information-restricted quantum mechanics.

The information-restricted quantum mechanics is one of the post-quantum theories that can also be contained in our framework.  In this theory, pure and pure stationary states have uniform state volumes but differ from $\kappa$ in $K(k)=\delta(k-\kappa)$. Consider modified quantum mechanics where all the pure states have the following form of the density matrix (in a particular basis):
\begin{equation}\frac{1}{2}\left[
\begin{array}{cccccc}
...\\
~&|a|^2&0&b^*a&0\\
~&0&|a|^2&0&b^*a\\
~&a^*b&0&|b|^2&0\\
~&0&a^*b&0&|b|^2\\
~&~&~&~&~&...
\end{array}\right],
\end{equation}
where $a,b$ are arbitrary complex numbers as long as the density matrix is normalized. Repeated numbers like $|a|^2,|b|^2$ are the restriction to density matrices (`$...$' contains $c,d,...$ terms in a similar form). In quantum mechanics, such a density matrix is a uniform mixture of two orthogonal pure states,
\begin{equation}\left[
\begin{array}{cccccc}
...\\
~&|a|^2&0&b^*a&0\\
~&0&0&0&0\\
~&a^*b&0&|b|^2&0\\
~&0&0&0&0\\
~&~&~&~&~&...
\end{array}\right],
\left[
\begin{array}{cccccc}
...\\
~&0&0&0&0\\
~&0&|a|^2&0&b^*a\\
~&0&0&0&0\\
~&0&a^*b&0&|b|^2\\
~&~&~&~&~&...
\end{array}\right].
\end{equation}(We can always find a basis where arbitrary two orthogonal states' density matrices have such a form.) Such a coarse-graining doubles the state volume of pure states, $V=2V_{quantum~pure~state}=2h$ by Theorem \ref{coarsegrain}. We assume the evolution is the same as in the standard quantum mechanics. Then, the state volumes of pure stationary states have similarly been doubled, while the non-localized dynamics kernel and equation of motion remain the same.

\subsection{Spekkens' toy model}
Beyond the phase space formalism, the idea of generalised energy eigenstates has broader application. Here we will discuss a discrete system Spekkens' toy model \cite{007SpeToy}, which has also been discussed in a GPT context \cite{017SpeToy,012SpeToy,013GPTPha}. Spekkens' toy model assumes four ontic states, labelled $1,2,3,4$. Additionally, there is an information restriction: all the measurements can only confirm a state is in one pair of ontic states or the other pair, for example, $1\vee 2$ or $3\vee 4$. Every measurement inevitably disturbs the ontic states, ensuring that consecutive measurements cannot gain better knowledge.
 Consequently, there are 6 pure epistemic states, $1\vee 2$, $1\vee 3$, $1\vee 4$, $2\vee 3$, $2\vee 4$, $3\vee 4$. The transformation of these epistemic states depends on the permutation of ontic states, which can be classified by:\\
1. All the ontic states are stationary. All the pure epistemic states are pure stationary states.\\
2. Two ontic states are stationary. For example, only $1,2$ are stationary, pure stationary states are $1\vee 2$, $3\vee 4$, $1\vee 3+1\vee 4$, $2\vee 3+2\vee 4$. ($+$ here represents a uniform probabilistic mixture. \footnote{Strictly speaking, we are discussing a generalised Spekkens' toy model with convex state space.})\\
3. One ontic state is stationary. For example, $1$ is stationary, pure stationary states are
$1\vee 2+1\vee 3+1\vee 4$, $2\vee 3+2\vee 4+3\vee 4$.\\
4. Permutation in pairs. For example, $1,2$ permute, $3,4$ permute, pure stationary states are $1\vee 2$, $3\vee 4$.\\
5. Cyclic permutation of four ontic states. The only pure stationary state is the maximally mixed state, $1\vee 2\vee 3\vee 4$.

Pure stationary states can determine the evolution except for the cyclic permutation of four ontic states. The only pure stationary state is the maximally mixed state. We cannot distinguish whether it is $1\mapsto2\mapsto3\mapsto4\mapsto1$ or $1\mapsto3\mapsto2\mapsto4\mapsto1$ or other possibilities. As we have mentioned, deriving an equation of motion from pure stationary states is generally a non-trivial task. For Spekkens' toy model, we need more conditions to derive a unique equation of motion.

Meanwhile, we can also find that the pure stationary states cannot provide a state-dual measurement. For example, the permutation of $1,2,3$ has the pure stationary states $1\vee 2+2\vee 3+3\vee 1$ and  $1\vee 4+2\vee 4+3\vee 4$, which are non-orthogonal, so there is no measurement can satisfy Definition \ref{dualself}. The model violates the Postulate \ref{measure}, measurements are so limited that a state-dual energy measurement does not always exist.

We can also generalise the effect/state volume to Spekkens' toy model. It still gives a proper normalization factor in measurement. Following Theorem \ref{defvolume} which gives volume, we can choose a form of the inner product to define the state volume. The dot product is a natural choice if we hope all the ontic states are on equal footing. (Other definitions of the inner product also work.) Then, the state volumes in Spekkens' toy model can be defined by:
\begin{equation}V_f=\frac{1}{\sum_{i=1}^4 f_i^2},
\end{equation} where $f_i$ represent the probabilities of the $i$-th ontic state. All the pure epistemic states share the same state volume 2, exactly the number of possible ontic states. Likewise, the state volume of the maximally mixed state is 4. These particular state volumes are the consequences of the inner product which depends on the symmetries we demand.
Finally, we test the corresponding state-dual measurement. Consider the probability from $1\vee2$ to $1\vee3$:

\begin{equation}
P=V_{1\vee3} \vec f_{1\vee2} \cdot \vec f_{1\vee3} =\frac{1}{2}.
\end{equation}

The effect volume which equals state volume gives the outcome as expected. The example of Spekkens' toy model shows that our framework also, at least partially, works on discrete systems.

\section{Summary and Outlook}

We built a framework describing a generalised energy concept and time evolution rule, which describes quantum and classical in a unified way. {We introduced 6 postulates: (1). Canonical symmetries; (2). Local inner product; (3). Pure stationary states decide the evolution; (4). Pure stationary states are independent; (5). Inner product invariance; (6). Existence of energy measurement. (1)-(2) Provide a unique inner product in phase space, which helps to define state-dual measurements. (1), (3)-(5) derive our generalised equation of motion Eq. \eqref{Commutator EOM} in phase space. Based on the above results, (6) further guarantees the existence of a state-dual measurement of the Hamiltonian, i.e., the conserved observable describing time evolution.  Rather than taking an algebraic approach, we endeavoured to make every postulate have a clear physical/operational meaning. We derived a generalised Hamiltonian system in phase space that encompasses quantum and classical theories but also generalises the original ideas.} This includes generalizing Planck's constant.

In our framework, the Planck constant provides the state/effect volume of pure states and corresponding measurements. It also appears in the equation of motion. In general, there is no good limit when taking a quantum state's volume to zero to get a classical state. Still, it is always possible to take the limit of the quantum equation of motion to get the corresponding classical evolution. The two roles are related. For example, physically non-localized evolution causes a negative distribution, so we need a non-zero state volume to avoid negative probability. However, they can have different values in general theories.
This framework possesses the potential for application and advancement in various directions. Firstly, there exists an intriguing connection between contextuality and evolution. Specifically, when $K(k)$ is not directly proportional to $ \delta(k)$, the evolution associated with infinite-order derivatives \cite{948Local,950NoLoAc} emerges alongside a negative distribution, which {could relate} to contextuality \cite{008NegCon}.

Secondly, it is possible to construct alternative forms of mechanics that do not conform to classical or quantum paradigms. One approach involves considering different values of the Planck constant in the equation of motion compared to its value in uncertainty. Another possibility entails selecting a non-delta function for $K(k)$, which can be interpreted as a probabilistic combination of diverse commutation relations during the process of evolution.

Thirdly, this framework facilitates clear analogies and comparisons between quantum and classical dynamics. For instance, it can be employed to elucidate the apparent acceleration exhibited by quantum walks in contrast to classical walks \cite{venegas2012quantum}.

Lastly, it is natural to utilize this framework to formulate a theory of thermodynamics that is not reliant on the underlying choice of mechanics.

\noindent {\bf {\em Acknowledgements.}} We thank Meng Fei, Chelvanniththilan Sivapalan, Andrew Garner and Dario Egloff for discussions.  LJ and OD acknowledge support from the National Natural Science Foundation of China (Grants No.12050410246, No.1200509, No.12050410245) and City University of Hong Kong (Project No. 9610623).

\appendix

\section{Some mathematical aspects of the phase space formalism}\label{phs}
In the settings we consider the $2n$ local coordinates $z=(q,p)$  on $\gP$
 are given by the generalized coordinates of the configuration space $q=\{q^a\}$ and the generalized momenta $p=\{p_a\}$, $a=1,\ldots, n$. These momenta are related to the coordinates $q$ and velocities $\dot q$ via $q_a\defeq\pad L/\pad \dot x^a$, where $L(q,\dot q)$ is the system's Lagrangian. As the system is unconstrained, the $n$ equations for momenta can be inverted to express the velocities as functions of positions and momenta.

$\gP$ is a symplectic manifold as
a non-degenerate closed 2-form is defined on it. It can always be written in local coordinates as
\be
\omega^{(2)}=dp_a\wedge dq^a\equiv dp\wedge dq,
\ee
and in these coordinates, it is represented as a $2n\times2n$ matrix
\be
 J= \begin{pmatrix}0 & I\\
-I & 0\end{pmatrix}.
\ee

The symplectic form establishes the isomorphism between vectors and 1-forms (covectors) on $\gP$ by matching a vector $\veta$ with the form $\omega^{(1)}_\veta$ via $\omega^{(1)}_\veta(\boldsymbol{\xi})\defeq\omega^{(2)}(\veta,\boldsymbol{\xi})$. Hence the Hamilton (canonical) equations of motion are written as
\be
\dot z(t)=\vH\big(z(t)\big),
\ee
representing the Hamiltonian phase flow
\be
\vH=J\nabla_z H.
\ee

The Poisson brackets of Eq.~\eqref{Leq} are
 Then
\be
 \{f,g\}=-\omega^{(2)}\big(\nabla_z f,\nabla_z g\big)=(\nabla_z f)^T\cdot J\cdot \nabla_z g.
\ee

Classical observables are smooth functions on the phase space and form the algebra. The Poisson bracket
can be defined more abstractly  as a Lie bracket on the underlying manifold: it is linear, anti-symmetric and satisfies the Jacobi identity
\begin{align}
 \{f,\{g,h\}\}=\{\{f,g\},h\}+\{g,\{f,h\}\}.  \label{l2}
\end{align}
In addition, it satisfies the Leibnitz rule with respect to the
product defining the algebra, $f\circ g (q,p)\defeq f(q,p)g(q,p)$,
\be
\{f,gh\}=\{f,g\}h+g\{f,h\}. \label{ll}
\ee
Technically this is the Jordan--Lee algebra with associative multiplication, i.e. the Poisson algebra.

\section{Wigner transport equation in quantum theory for general Hamiltonian} \label{deltatrick}
In this appendix, we describe how to convert between different expressions of the time evolution in the quantum case. The original Wigner transportation function Eq. (\ref{wignereq}) only provides jumping in the momentum direction. In that case, $J(q,j)$ can be viewed as a Fourier transform of the odd part of the potential, so strictly speaking, it is well-defined only for restricted Hamiltonian. However, if you accept derivatives of delta functions as the Fourier transform of polynomials, this formula can give a reasonable description for general Hamiltonians.

\subsection{Case of $H=p^2/2m$}
As an example, we can first generalise this idea to the kinetic term of Hamiltonian $H=p^2/2m$. We attempt to find an apparently non-localized expression that
\begin{equation}\frac{\partial W}{\partial t}=\int W(q+j,p)J_q(p,j)dj=-\frac{p}{m}\pdv{W}{q}.\end{equation}
(It is also the equation of motion for classical free particles.) The delta function's derivatives can be defined through partial integral,
\begin{equation}
\int \delta'(x-x_0)f(x)dx=-f'(x_0).
\end{equation}
Therefore, the kinetic effect can be expressed by
\begin{equation}\pdv{W}{t}=\frac{p}{m}\int W(q+l,p) \delta'(l)dj.\end{equation}

After this, we check if the Wigner transportation function also works for the kinetic term:
\begin{equation}J_q(p,l)=\frac{p}{m}\delta'(l)=\frac{p}{2\pi m}\int(-2iy/\hbar) e^{-2ily/\hbar}d(-2y/\hbar)\end{equation}
\begin{equation}=\frac{-i}{2\pi\hbar}\int \frac{4py}{2m} e^{-2ily/\hbar} d(-2y/\hbar)\end{equation}
\begin{equation}=\frac{i}{\pi\hbar^2} \int [T(p+y)-T(p-y) ]e^{-2ily/\hbar} dy\end{equation}
$$(\delta(x)=\frac{1}{2\pi}\int e^{ikx}dk,~\delta'(x)=\frac{1}{2\pi}\int ik e^{ikx}dk)$$
Therefore, we can see that the Wigner transportation function also works for the power series once we introduce the derivatives of the delta function as the Fourier transform of the power series.

\subsection{General $H$}
Based on the above idea, we check the generalised apparently non-localized Eq. (\ref{GWignerequation}) to see if it can convert to the original apparently localized Eq. (\ref{WFEOM}).
Use the Taylor expansion of multivariate function to expand $H(q-y,p+z)-H(q+y,p-z)$ terms in Eq. (\ref{GWignerequation}):
\begin{widetext}
\begin{equation}H(q-y,p+z)-H(q+y,p-z)=-2y\frac{\partial H}{\partial q}+2z\frac{\partial H}{\partial p}+... =\sum_{m,n|m,n\in \mathbbm{N}, m+n\in \text{odd}}2((-y)^mz^n)\frac{C^m_{m+n}}{(m+n)!}\frac{\partial^{m+n} H}{\partial q^m \partial p^n}.\end{equation}
\begin{equation}J=\sum_{m,n|m,n\in \mathbbm{N}, m+n\in \text{odd}}2\frac{i}{\hbar}(\frac{i\hbar}{2})^{m+n}(-1)^m\delta^{(m)}(j)\delta^{(n)}(l)\frac{C^m_{m+n}}{(m+n)!}\frac{\partial^{m+n} H}{\partial q^m \partial p^n}.\end{equation}
\begin{equation}\frac{\partial W}{\partial t}=\sum_{m,n|m,n\in \mathbbm{N}, m+n\in \text{odd}}2\frac{i}{\hbar}(\frac{i\hbar}{2})^{m+n}(-1)^m(-1)^{m+n} \frac{\partial^{n+m} W}{\partial q^n\partial p^m}\frac{C^m_{m+n}}{(m+n)!}\frac{\partial^{m+n} H}{\partial q^m \partial p^n},\end{equation}
\begin{equation}\frac{\partial W}{\partial t}=\frac{-2}{\hbar}\sum_{m,n|m,n\in \mathbbm{N}, m+n\in \text{odd}}(\frac{\hbar}{2})^{m+n}(-1)^{\frac{m+n+1}{2}}(-1)^m\frac{\partial^{n+m} W}{\partial q^n\partial p^m}\frac{C^m_{m+n}}{(m+n)!}\frac{\partial^{m+n} H}{\partial q^m \partial p^n}.\end{equation}
On the other side, expand Eq. (\ref{WFEOM})
\begin{equation}\frac{\partial W}{\partial t}=\frac{-2}{\hbar}H \sin(\frac{\hbar\Lambda}{2})W(q,p)=\frac{-2}{\hbar}\sum_{a,b|a,b\in \mathbbm{N}, a+b\in \text{odd}} \frac{C^a_{a+b}}{(a+b)!}(\frac{\hbar}{2})^{a+b} (-1)^{\frac{a+b-1}{2}}(-1)^b\frac{\partial^{a+b}H}{\partial q^b\partial p^a}  \frac{\partial^{a+b} W}{\partial q^a\partial p^b} \end{equation}
\end{widetext}

Notice that $m+n$ is odd, so $(-1)^m=-(-1)^n$, two equations (\ref{WFEOM}), (\ref{GWignerequation}) are completely equivalent.
\section{Symmetry in the equation of motion}\label{symmetry}
We first take time-reversal symmetry as an example to show how it restricts $J^0$. The reversal symmetry operation changes the following elements: $\pdv{f'}{t}\text{}(0,0)=-\pdv{f}{t}\text{}(0,0),$ $g'(q,p)=g(q,-p)$, $f'(l,j)=f(l,-j)$. Substitute them into the Eq. (\ref{exofsym}), $\pdv{f'}{t'} \text{}(0,0)=\int {J^0}'_{g'}(l,j) f'(l,j)dldj:$
\begin{equation}-\pdv{f}{t}\text{} (0,0)=\int {J^0}_{g'}(l,j) f(l,-j)dldj,\end{equation}
\begin{equation}\pdv{f}{t}\text{} (0,0)=\int (-{J^0}_{g'}(l,-j)) f(l,j)dldj.\end{equation}
Since $f$ can be an arbitrary function,
\begin{equation}-{J^0}_{g'}(l,-j)={J^0}_{g}(l,j). \end{equation}
The symmetry operation has transformed $g\rightarrow g'$. We can decompose $g$ by eigenstates of the operation $g_1(q,p)=g_1(q,-p)$ and $g_2(q,p)=-g_2(q,-p)$, and the corresponding $J^0$ has the symmetry:
\begin{equation}{J^0}_{g_1}(l,-j)=-{J^0}_{g_1}(l,j),~{J^0}_{g_2}(l,-j)={J^0}_{g_2}(l,j).\end{equation}
We have learned the symmetry of $J$ corresponding to eigenstates of symmetry operation.

We consider three symmetry operations in total: time reversal, parity (the composition of time reversals and switches, $(q,p,t)\mapsto(-q,-p,t)$), and switch+time reversal+switch $(q,p,t)\mapsto(-q,p,-t)$. We decompose pure stationary states $g$ by joint eigenstates of symmetry operations with eigenvalue $\pm1$. First, even and odd functions (eigenstates of parity):
\begin{equation}\begin{aligned}
 g = &g_\text{odd}+ g_\text{even},\\
 \text{where}~~~~~~~~~~&\\
 g_\text{odd}(a,b)&= -g_\text{odd}(-a,-b),\\
 g_\text{even}(a,b)&= g_\text{even}(-a,-b).
 \end{aligned}
\end{equation}
 We can further decompose the odd and even parts by the joint eigenstates of parity, time reversal, and switch+time reversal+switch. The even part can be decomposed by delta functions:
\begin{equation}\begin{aligned}
 g_\text{even}&=\int \rho_1(q_0,p_0) g_{e1}+\rho_2(q_0,p_0) g_{e2}dq_0dp_0,\\
 \text{where}~~~~~~~~~~&\\
 g_{e1}(q,p,q_0,p_0)&=\delta(q-q_0)\delta(p-p_0)+\delta(q-q_0)\delta(p+p_0)\\
 &~~~~+\delta(q+q_0)\delta(p-p_0)+\delta(q+q_0)\delta(p+p_0),\\
 g_{e2}(q,p,q_0,p_0)&=\delta(q-q_0)\delta(p-p_0)-\delta(q-q_0)\delta(p+p_0)\\
 &~~~~-\delta(q+q_0)\delta(p-p_0)+\delta(q+q_0)\delta(p+p_0),
 \end{aligned}
\end{equation}$\rho_1$ and $\rho_2$ are corresponding weight of $g_{e1}$ and $g_{e2}$. The odd part can be decomposed by similar delta functions. Fig. \ref{evenpss} illustrates the eigenstates of parity, time reversal, and switch+time reversal+switch symmetry, and how their $J^0$ is restricted from symmetries.

\begin{figure}
  \centering
  \includegraphics[width=7cm,trim=90 1 120 0,clip]{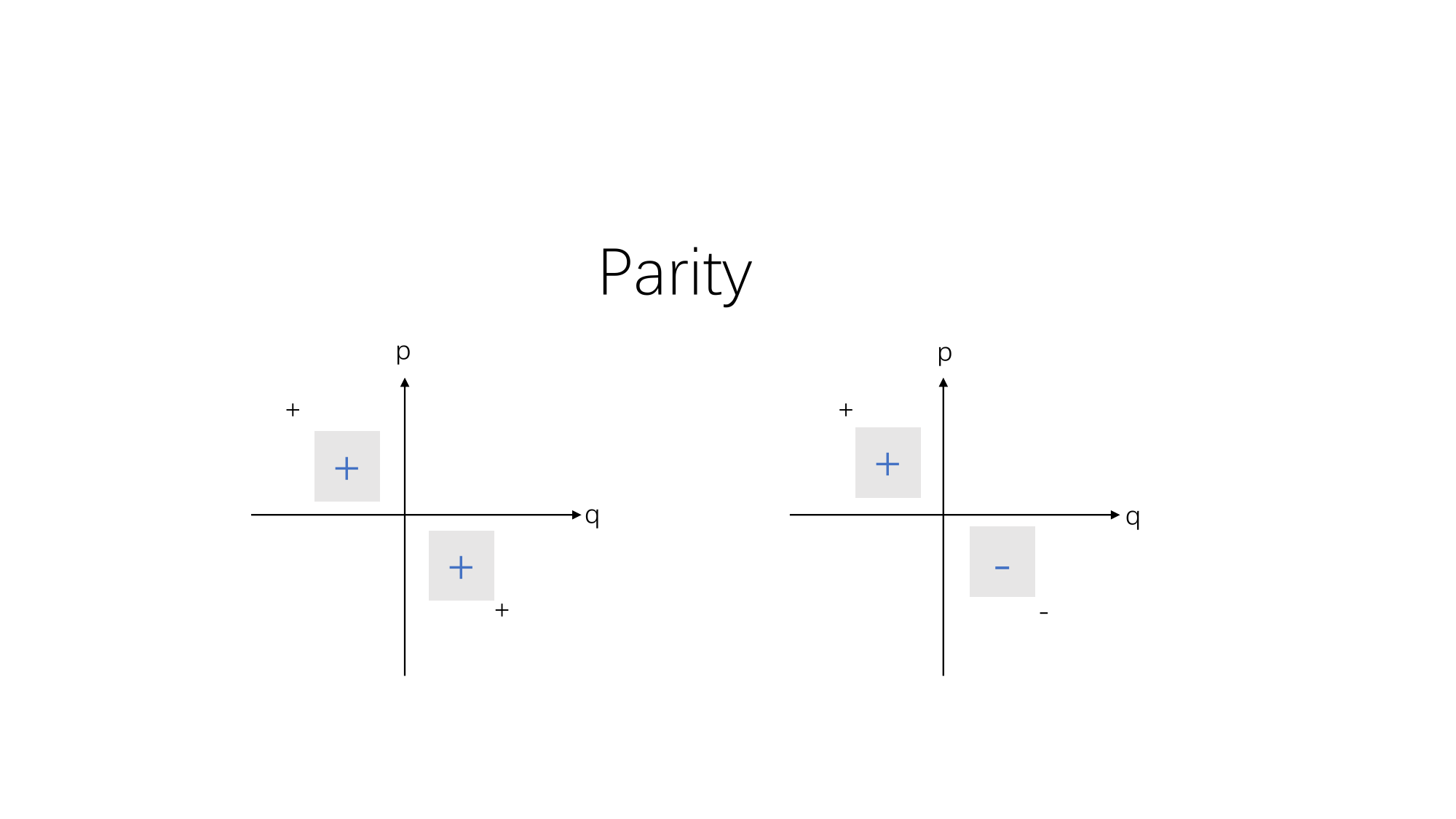}
  \includegraphics[width=7cm,trim=90 1 180 0,clip]{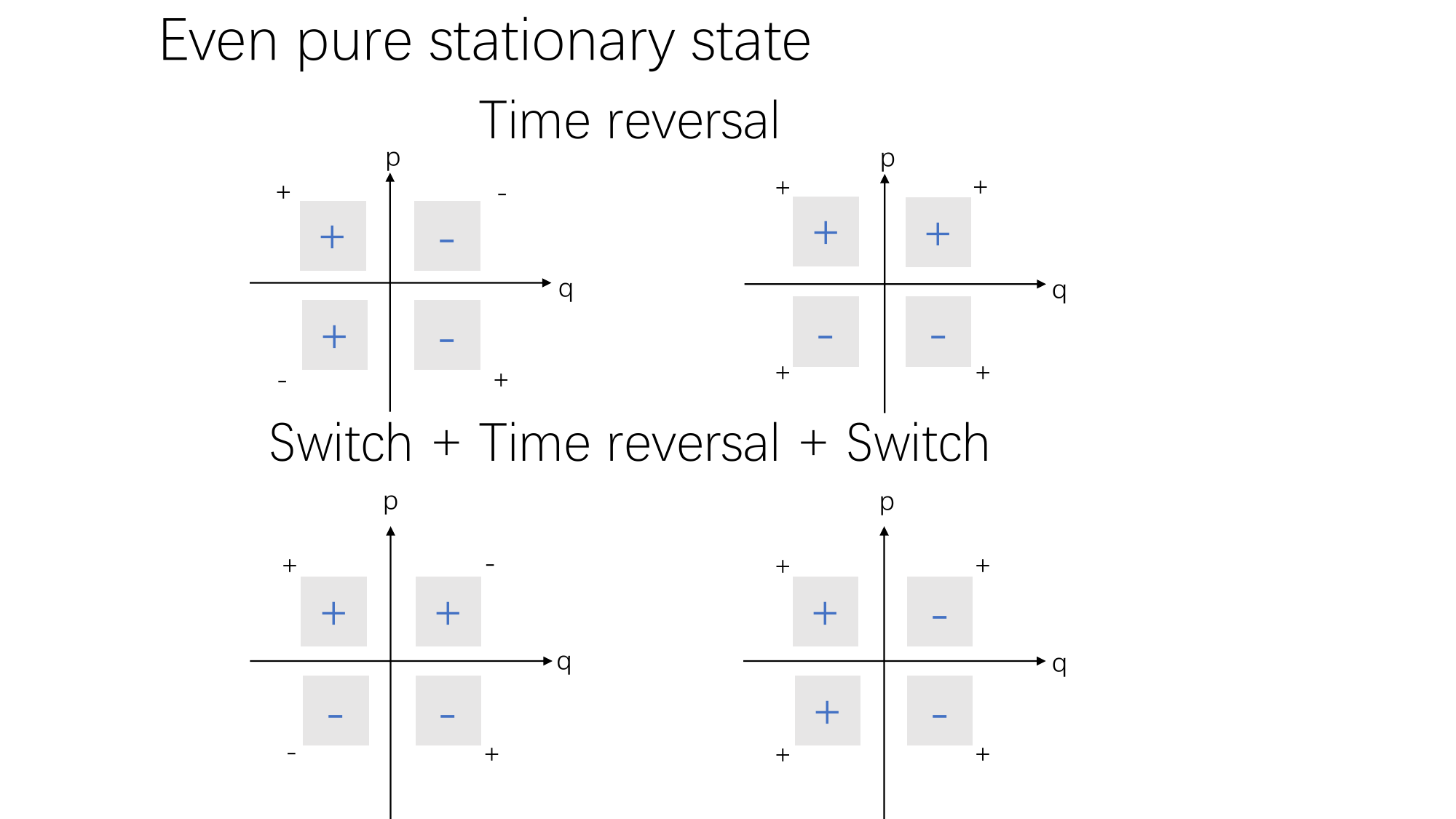}
  \includegraphics[width=7cm,trim=90 1 180 0,clip]{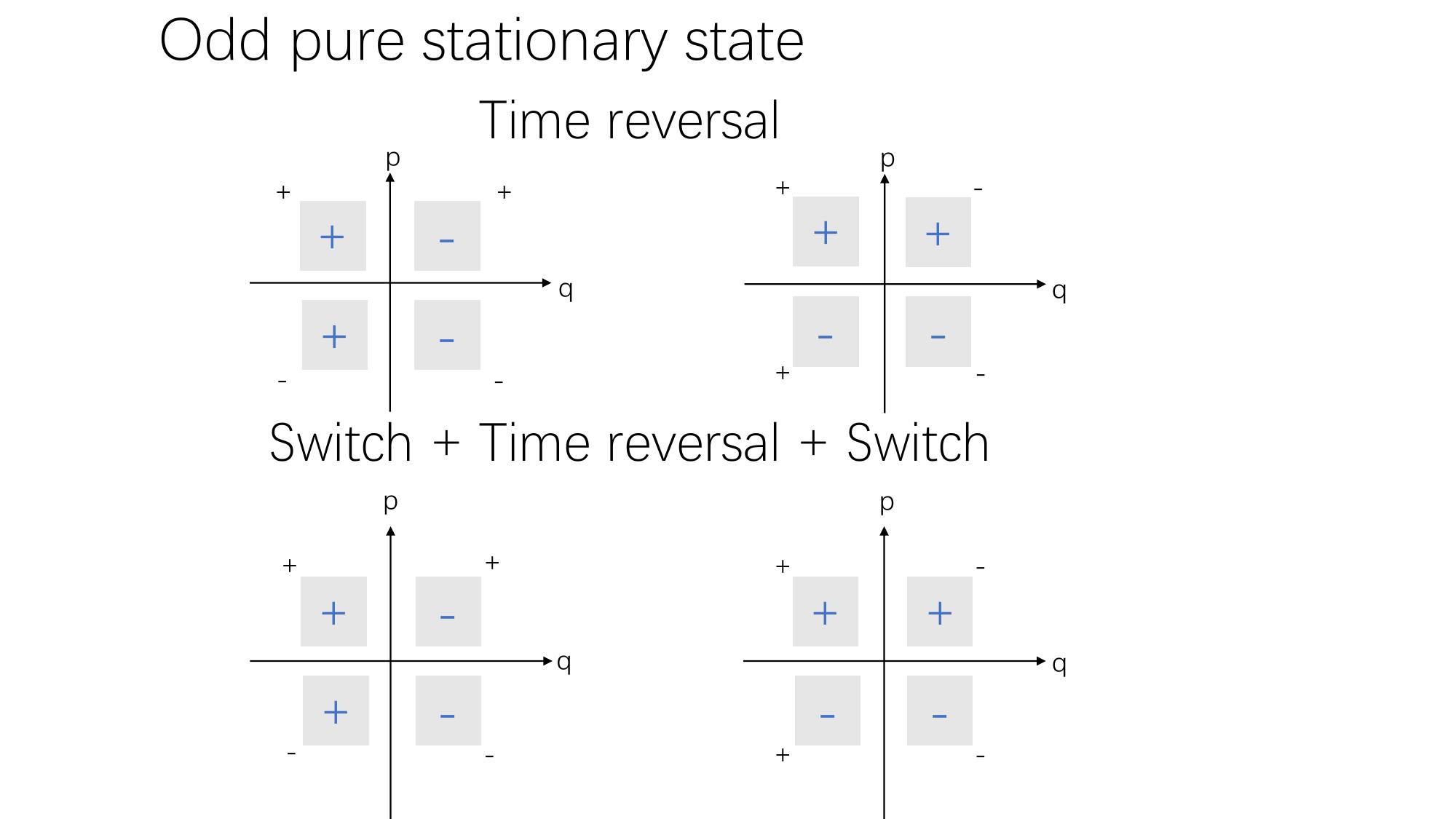}
  \caption{The black +,- notations donate the delta functions as components of the pure stationary state, and blue +,- with grey shadow donate the required symmetry of ${J^0}(l,j)$. The first picture shows the symmetry of $\tilde{J}$ required by parity, and the following pictures show the requirement by time reversal and switch+time reversal+switch given even/odd pure stationary states. The left and right figures show different eigenstates with eigenvalue $\pm 1$, respectively. We find for an even pure stationary state, the time-reversal symmetry contradicts the switch+time-reversal+switch symmetry unless ${J^0}_{g_{even}}=0$.}\label{evenpss}
\end{figure}
For an even pure stationary state, the time-reversal symmetry contradicts the switch+time-reversal+switch symmetry unless ${J^0}_{g_{even}}=0$. For an odd pure stationary state, the requirements are consistent. ${J^0}_{g_{odd}}(l,j)$ only has to be an odd function due to the parity symmetry.

Therefore, only the odd part of the pure stationary state contributes to the $J$, and
\begin{equation}
{J^0}(l,j)=-{J^0}(-l,-j)
\end{equation} for arbitrary pure stationary states.
We get lemma \ref{symlem}.

\section{The operator form of the equation of motion}\label{suan}
The Wigner transform is:
\begin{equation}A(q,p)=\int dz e^{i\frac{2}{\hbar}pz}\expval{q-z|\hat A|q+z}.\end{equation}
However, we can use an arbitrary factor labelled by $\mu$ instead of $\hbar$ to do the Wigner transform. Consider the following expression,
\begin{equation}\int A(q+l,p+j)B(q+y,p+z)K(k)e^{-i\frac{2}{k}(jy-lz)}dkdydzdldj,\end{equation}
 we replace phase space functions $A$ and $B$ by their operator forms:
 \begin{widetext}
 \begin{align}=&\int \expval{q+l-m|\hat A_{\mu}|q+l+m}e^{im\frac{2}{\mu}(p+j)} \expval{q+y-n|\hat B_{\mu}|q+y+n}e^{in\frac{2}{\mu}(p+z)}K(k) e^{-i\frac{2}{k}(jy-lz)}dkdydzdldj dndm\\
 =&\int \expval{q+l-m|\hat A_{\mu}|q+l+m} \expval{q+y-n|\hat B_{\mu}|q+y+n}e^{i\frac{2}{\mu}(m+n)p} K(k) e^{i(\frac{2}{\mu}m-\frac{2}{k}y)j}e^{i(\frac{2}{\mu}n+\frac{2}{k}l)z} dkdydzdldj dndm\\
=&4\pi^2\int \expval{q+l-m|\hat A_{\mu}|q+l+m} \expval{q+y-n|\hat B_{\mu}|q+y+n}e^{i\frac{2}{\mu}(m+n)p} K(k) \delta\bigg(\frac{2}{\mu}m-\frac{2}{k}y\bigg)\delta\bigg(\frac{2}{\mu}n+\frac{2}{k}l\bigg) dkdydldndm\\
=&4\pi^2\int \expval{q-\frac{kn}{\mu}-m|\hat A_{\mu}|q-\frac{kn}{\mu}+m} \expval{q+\frac{km}{\mu}-n|\hat B_{\mu}|q+\frac{km}{\mu}+n}e^{i\frac{2}{\mu}(m+n)p}K(k)\bigg(\frac{k}{2}\bigg)^2 dkdndm.
\end{align}
The notations like $\hat A_{\mu}$ mean $A=\int dz e^{i\frac{2}{\mu}{pz}}\expval{r-z/2|\hat A_\mu|r+z/2},$ $\mu$ take the place of $\hbar$ in the original Weyl transform.
  To the same phase space function, different $\mu$ lead to different operators. We are free to choose the value of $\mu$, let $\mu=k$:
\begin{equation}
=4\pi^2\int \expval{q-n-m|\hat A_{k}|q-n+m} \expval{q+m-n|\hat B_k|q+m+n}e^{i\frac{2}{k}(m+n)p}K(k)\big(\frac{k}{2}\big)^2dmdndk.
\end{equation}
Let $m+n=a, m-n=b$, and then replace $q+b$ by $r'$:
\begin{equation}
=4\pi^2\int \expval{q-a|\hat A_k|q+b} \expval{q+b|\hat B_{k}|q+a}e^{i\frac{2}{k}bp}K(k)\bigg(\frac{k}{2}\bigg)^2dadbdk
=4\pi^2\int \expval{q-a|\hat A_k|r'} \expval{r'|\hat B_{k}|q+a}e^{i\frac{2}{k}bp}K(k)\bigg(\frac{k}{2}\bigg)^2dadr'dk.
\end{equation}
Compare the result with the Wigner transform of the product $\hat A\hat B$:
\begin{equation}AB=\int dzdr' e^{i\frac{2}{\hbar}pz}\expval{q-z|\hat A|r'}\expval{r'|\hat B|q+z},\end{equation}
we find:
\begin{equation}
\int A(q+l,p+j)B(q+y,p+z)K(k)e^{-i\frac{2}{k}(jy-lz)}dkdydzdldj= \pi^2\int k^2 K(k) \text{Wigner}_{k}\{\hat A_k \hat B_k\} dk.
\end{equation}
The LHS of the equation $\eqref{Commutator EOM}$ only takes the imaginary part, which can be computed by $\text{Im}f=\frac{f-f^*}{2i}$. Therefore, the overall result leads to a commutator-like equation, where the commutation relation is given by the distribution $K(k)$:
\begin{equation}
\begin{array}{rl}
&\text{Im} \int f(q+l,p+j) g(q+y,p+z)  K( k) e^{-i\frac{2(jy-lz)}{ k}} dkdydzdldj\\
=&\frac{1}{2i}\int [f(q+l,p+j) g(q+y,p+z)-f(q+y,p+z) g_i(q+l,p+j)] K( k) e^{-i\frac{2(jy-lz)}{ k}} dkdydzdldj\\
=&-i\frac{\pi^2}{2}\int k^2 K(k) \text{Wigner}_{k}\{[\hat f_{ k},\hat g_{ k} ]\} d k.
\end{array}
\end{equation}

\section{The associative Moyal product and the non-associative hybrid Moyal product}
\label{associate}
Observe that $$\partial_{x}\Big(f(x,y)|_{x=y}\Big) = \Bigg(\Big(\partial_ x+\partial_ y\Big)f(x,y)\Bigg)\Bigg|_{x=y}.$$

We introduce an notation $E(\partial_{x}, \partial_{y})f(x)g(y)|_{x=y}:= f(q,p) e^{i \frac{k}{2} \Lambda} g(q,p)$, where $x,y$ represent the vector $(q,p)$.

We consider an example of the hybrid Moyal product, $K(k)=\delta(k-k_1)+\delta(k-k_2)$, $f\star_H g:=E_1(\partial_{x}, \partial_{y})+E_2(\partial_{x}, \partial_{y}) f(x)g(y).$

Compute
\begin{align}
(f\star_H g)\star_H h&=\big[E_1(\partial_x,\partial_z)+E_2(\partial_ x,\partial_z)\big] \bigg\{\big[E_1(\partial_x,\partial_y)+E_2(\partial_x,\partial_y)\big]f(x)g(y)|_{x=y}\bigg\}h(z)|{x=z}\\
&=\big[E_1(\partial_x+\partial_y,\partial_z)+E_2(\partial_x+\partial_y,\partial_ z)\big] \big[E_1(\partial_x,\partial_ y)+E_2(\partial_x,\partial_ y)\big]f(x)g(y)h(z)|_{x=y=z}\\
&=\big[E_1(\partial_ x,\partial_ z) E_1(\partial_y,\partial_ z)+E_2(\partial_x,\partial_ z) E_2(\partial_ y,\partial_ z) \big] \big[E_1(\partial_ x,\partial_ y)+E_2(\partial_ x,\partial_ y)\big]f(x)g(y)h(z)|_{x=y=z}\\
&=\big[E_1(\partial_ x,\partial_ z) E_1(\partial_y,\partial_ z)E_1(\partial_ x,\partial_ y)+
E_2(\partial_ x,\partial_ z) E_2(\partial_y,\partial_ z)E_2(\partial_ x,\partial_ y)\\
&~~~~+E_1(\partial_ x,\partial_ z) E_1(\partial_y,\partial_ z)E_2(\partial_ x,\partial_ y)+E_2(\partial_ x,\partial_ z) E_2(\partial_y,\partial_ z)E_1(\partial_ x,\partial_ y)\big]f(x)g(y)h(z)|_{x=y=z}.
\end{align}

Similarly,
\begin{align}
f\star_H (g\star_H h)&=\big[E_1(\partial_ x,\partial_ z) E_1(\partial_y,\partial_ z)E_1(\partial_ x,\partial_ y)+
E_2(\partial_ x,\partial_ z) E_2(\partial_y,\partial_ z)E_2(\partial_ x,\partial_ y)\\
&~~~~+E_1(\partial_ x,\partial_ y) E_1(\partial_x,\partial_ z)E_2(\partial_ y,\partial_ z)+E_2(\partial_ x,\partial_ y) E_2(\partial_x,\partial_ z)E_1(\partial_ y,\partial_ z)\big]f(x)g(y)h(z)|_{x=y=z}.
\end{align}

If $E_1\neq 0$, $E_2=0$, namely a conventional Moyal product, we find the two results coincide, so the Moyal product is associative. However, the hybrid Moyal product is not associative in general.

\end{widetext}

\end{document}